%% file: automation_entrepreneurship_4_21_2026.tex
\documentclass[a4paper,12pt, oneside]{article}
\usepackage{newtxtext}
\usepackage{amsmath}
\usepackage{newtxmath}

\usepackage{amsthm}
\usepackage{microtype}
\usepackage[margin=1in]{geometry}
\usepackage{setspace}
\usepackage{fancyhdr}
\usepackage{appendix} 
\usepackage{graphicx}
\usepackage{float}
\usepackage{fp}
\usepackage{rotating}
\usepackage[para,online,flushleft]{threeparttable}
\usepackage{tabulary}
\usepackage{tabularx}
\usepackage{natbib}
\setcitestyle{authoryear,open={(},close={)}}
\usepackage[utf8]{inputenc}
\usepackage[english]{babel}
\usepackage{nicefrac}
\usepackage{placeins}
\usepackage{dcolumn}
\usepackage[mathcal]{eucal}
\usepackage{enumerate}
\usepackage{caption}
\DeclareCaptionFormat{twoline}{\textbf{#1#2}\\[3pt]\textbf{#3}}
\DeclareCaptionFormat{twoline_table}{\textbf{#1#2}\\[3pt]\textbf{#3}}
\usepackage{subcaption}
\usepackage{lscape}
\usepackage{multirow}
\usepackage{paralist}
\usepackage[affil-it]{authblk}
\usepackage[dvipsnames,svgnames]{xcolor}
\usepackage{booktabs}
\usepackage{threeparttablex}
\usepackage{longtable}
\usepackage{pdflscape}
\setTableNoteFont{\fontsize{8}{9.6}\selectfont}
\usepackage{titlesec}
\usepackage{afterpage}
\usepackage{url}
\usepackage{csquotes}
\usepackage{mathtools}
\usepackage{array}
\usepackage{accents}

\newcolumntype{x}[1]{>{\centering\let\newline\\\arraybackslash\hspace{0pt}}p{#1}}
\usepackage{tikz}
\usepackage{textcomp}
\usetikzlibrary{positioning}
\usetikzlibrary{shapes,arrows}
\usepackage{siunitx}
\newcolumntype{d}[1]{D{.}{.}{#1}}
\definecolor{blue}{rgb}{0,0.08,0.5}
\definecolor{red}{rgb}{.6,0,0}
\definecolor{green}{rgb}{0,0.376,0}
\usepackage[bookmarks, pdftitle={}, pdfauthor={}, colorlinks=true, linkcolor=red, citecolor=blue, urlcolor=blue]{hyperref} 
\usepackage{xparse}
\usepackage[hang,flushmargin]{footmisc}
\usepackage[hang,flushmargin]{footmisc}

\ExplSyntaxOn
\NewDocumentCommand{\longdash}{ O{2} }
 {
  --\prg_replicate:nn { #1 - 1 } { \negthinspace -- }
 }
\ExplSyntaxOff

\newcommand{\Expect}{{\rm I\kern-.3em E}}
\newsavebox\mytabularbox
\newcount\figwidthc
\newcount\textwidthc
\newcolumntype{Y}{>{\centering\arraybackslash}X}
\newcommand\totextwidth[1]{%
  \sbox{\mytabularbox}{#1}%
  \figwidthc=\wd\mytabularbox%
  \textwidthc=\textwidth%
  \FPdiv\scaleratio{\the\textwidthc}{\the\figwidthc}%
  \FPmin\scaleratio{\scaleratio}{1}%
  \scalebox{\scaleratio}{\usebox{\mytabularbox}}%
}
\newtheorem{proposition}{Proposition}

\begin{document}   
\pagestyle{empty}

\title{\textbf{Routine Work, Firm Boundaries, and \\ the Rise of Local Supplier Entry}}

\author{
\textnormal{Duha T. Altindag, Nabamita Dutta, John M. Nunley,}\\
\textnormal{R. Alan Seals and Adam Stivers}
\thanks{Altindag (\href{mailto:altindag@auburn.edu}{altindag@auburn.edu}) and Seals (\href{mailto:alan.seals@auburn.edu}{alan.seals@auburn.edu}) are affiliated with Auburn University, Dutta (\href{mailto:ndutta@uwlax.edu}{ndutta@uwlax.edu}), Nunley (\href{mailto:jnunley@uwlax.edu}{jnunley@uwlax.edu}), and Stivers (\href{mailto:astivers@uwlax.edu}{astivers@uwlax.edu}) with University of Wisconsin--La Crosse.}
}

\date{\today}

\maketitle

\thispagestyle{empty}

\begin{abstract}
\singlespacing
\noindent Between 2005 and 2019, U.S.\ business applications rose 40 percent while conversion to employer firms fell by nearly half. We study whether boundary redrawing helps explain this pattern. Structured routine-cognitive work can be governed through deliverables and thinner buyer and supplier interfaces. When such work remains place-bound, outsourcing creates demand for domestic specialist suppliers. Across 722 commuting zones, a one percentage-point higher baseline routine employment share raises applications by 27.8 per 100{,}000 residents. Realized entry concentrates in micro-establishments, with no startup quality gains. Contract and industry evidence point to local supplier entry, not routine-manual displacement.

\vspace{1em}

\vspace{0.5em}
\end{abstract}

\pagebreak

\section*{Extended Abstract}
\singlespacing

\noindent Between 2005 and 2019, U.S.\ business applications rose sharply while conversion to employer firms weakened. We study whether part of this pattern reflects boundary redrawing by incumbent firms. Our mechanism is interface compression. Structured routine work can be separated from the firm when it can be described, monitored, and evaluated through deliverables. When these activities remain tied to buyers, workers, facilities, or customers, outsourcing can create demand for domestic specialist suppliers.

\noindent We test this idea across 722 U.S.\ commuting zones. Our key exposure is baseline Routine Employment Share, \(RSH\), measured in 2005 as the share of local employment in occupations above the 66th percentile of the routine-task distribution in \citet{deming_growing_2017}. We interpret \(RSH\) as a proxy for structured, codifiable work near the interface-compression margin. Our primary outcome is business applications from the Census Bureau's Business Formation Statistics. We instrument \(RSH\) using a 1990 leave-one-state-out shift-share measure based on historical industry composition and national industry routine employment shares.

\noindent A one percentage-point increase in baseline \(RSH\) raises business applications by 27.8 per 100{,}000 residents from 2005 to 2019. A one standard deviation difference in \(RSH\) corresponds to roughly 86 additional applications per 100{,}000 residents. The result is stable across alternative base years, controls, functional forms, inference checks, and Rotemberg-weight diagnostics. The response is driven by routine-cognitive work rather than routine-manual work, and it remains after conditioning on instrumented industrial-robot exposure.

\noindent The application response is not only filing behavior. Business Dynamics Statistics show that realized establishment entry also moves in the same direction, although by less than applications. The realized-entry response appears after 2010 and is concentrated among micro-establishments with 1 to 4 employees. Startup Cartography data show no improvement in startup quality or quality adjusted entrepreneurial capital. This matches the mechanism. Boundary redrawing predicts more entry on the extensive margin, especially among small local suppliers, not necessarily better startups.

\noindent We then examine the organization of production. Higher routine exposure predicts more active four-digit industries and lower secondary production. Local production becomes broader across industries while production profiles become narrower. SEC Exhibit 10 filings point in the same direction. High-\(RSH\) commuting zones are more exposed to industries whose buyer and supplier contracts shift toward performance, deliverables, and codified buyer governance. Credit-agreement language shows no analogous shift.

\noindent The evidence supports a domestic supplier-entry margin in the adjustment to routine and codifiable work. Prior research shows that routine tasks are easier to source at arm's length across borders. We show that this is not the only margin. When separable activities remain place-bound, routine-cognitive structure can create demand for domestic specialist suppliers. Some new firms arise because incumbents redraw their boundaries and create markets for specialized suppliers.

\pagebreak
\pagestyle{plain}
\setcounter{page}{1}
\onehalfspacing

%%%%%%%%%%%%%%%%%%%%%%%%%%    INTRO %%%%%%%%%%%%%%%%%%%%%%%%%%%%%%%%%%%
\section{Introduction} \label{section:intro}

Routine work is often viewed as tasks that can be moved away from incumbent firms. In the automation literature, its rule-based structure makes it vulnerable to machines and software. Industrial robots have been shown to reduce employment and wages in exposed local labor markets, especially in routine manual production \citep{acemoglu_robots_2020}. In the offshoring literature, the same structure makes routine tasks easier to govern at a distance because they generate fewer unforeseen contingencies and are cheaper to monitor at arm's length \citep{grossman_trading_2008, costinot_adaptation_2011, oldenski_task_2012}. We show that this is only part of the adjustment. Places with more routine work, especially routine-cognitive work, also generate more domestic business applications and micro-establishment entry.

The key distinction is that organizational separation need not imply geographic relocation. Routineness can make an activity easier to separate from the incumbent firm because the work can be described, monitored, and evaluated through deliverables. But many activities that can be separated organizationally still require local delivery. Staffing, facilities support, on-site administrative work, logistics coordination, local compliance, and some health-system support services depend on proximity to buyers, workers, facilities, or local customers. When incumbents outsource these activities, the work may leave the buyer's payroll while remaining in the local economy and creating demand for nearby suppliers. We call this margin \emph{interface compression}, the replacement of internal supervision with a thinner buyer--supplier interface. The expected footprint of this is more local, small supplier firms rather than better or faster-growing startups.

We take this prediction to data using 722 U.S.\ commuting zones from 2005 to 2019. Our key measure is the baseline \emph{Routine Employment Share} (\(RSH\)), defined as the share of local employment in occupations above the 66th percentile of the routine-task distribution in \citet{deming_growing_2017} and measured in 2005. We interpret \(RSH\) as a measure of local exposure to structured, standardized work. Our primary outcome is business applications from the Census Bureau's Business Formation Statistics, which capture entry intent rather than realized employer formation. We relate long differences in business applications per 100{,}000 residents to baseline \(RSH\). Because local occupational composition may be measured with error and correlated with persistent regional conditions, we instrument \(RSH\) with a leave-one-state-out shift-share measure built from 1990 industry composition and national industry routine employment shares.

The entry response is economically meaningful and selective. Our 2SLS estimates imply that a one-standard-deviation difference in \(RSH\), equal to 3.1 percentage points, corresponds to roughly 86 additional business applications per 100{,}000 residents. The effect is driven by routine cognitive rather than routine manual work. Higher-\(RSH\) places also record more realized establishment births, and a shift toward micro-establishments with 1--4 employees. There are no discernible improvements in startup quality. Pre-2005 data show no stable positive relationship between baseline routine intensity and establishment entry. The positive association emerges only after 2005 and strengthens after 2010. Conditioning on instrumented industrial-robot exposure, which has its own independent effect, leaves the coefficient on \(RSH\) largely unchanged \citep{acemoglu_robots_2020}. The evidence, therefore, points away from a generic routine-work displacement story and toward a routine-cognitive supplier-entry margin.

We then ask whether the local entry response is accompanied by changes in the organization of production. At the commuting-zone level, higher routine-cognitive exposure predicts more active four-digit industries and less secondary production. These patterns are consistent with incumbents narrowing their own activity bundles while specialized suppliers absorb peripheral functions. We also use SEC Exhibit 10 filings to examine buyer--supplier contract language. High-\(RSH\) commuting zones are more exposed to industries whose business-to-business contracts shifted toward performance-and-deliverables language, automation-and-technology language, and a broader composite of outsourcing-oriented features. Analogous shifts in credit-agreement language are absent. Taken together, the evidence is consistent with boundary redrawing and supplier entry rather than with a pure worker-displacement story.

This domestic supplier-entry margin is the paper's central contribution. Prior work shows that routine and codifiable tasks are easier to source at arm's length across borders, so routineness can facilitate offshoring and global reorganization \citep{grossman_trading_2008, costinot_adaptation_2011, oldenski_task_2012, oldenski_offshoring_2014}. For activities that remain place-bound, the same routine task structure can also predict a domestic supplier response. High-routine commuting zones generate more local business applications and micro-establishment entry, and they exhibit broader industry presence and contract patterns consistent with suppliers absorbing activities that incumbents shed. Greater arm’s-length governability need not require all separated work to be moved out of the United States. It can also reorganize production within the United States by creating demand for domestic supplier firms.

This interpretation connects to research on firm boundaries and fissuring. The firm-boundaries literature makes contractibility central to the choice between integration and outsourcing \citep{grossman_costs_1986, hart_property_1990, antras_firms_2003, baker_contractibility_2004}. We do not observe contractibility directly. We use routine-task structure as a proxy for standardized work, which may be easier to govern through deliverables, and test whether such work predicts local supplier entry. The fissuring literature emphasizes that lead firms retain control over brands, standards, and customer relationships while shifting employment to contractors, franchisees, staffing firms, and other intermediaries \citep{weil_fissured_2014, weil_understanding_2019, bernhardt_domestic_2016}. That literature has mainly focused on wages, benefits, and job quality \citep{dube_does_2010, goldschmidt_rise_2017}. We share its emphasis on boundary redrawing, but our focus is on the domestic business-formation margin.

Our paper also contributes to entrepreneurship research by emphasizing a demand-side source of entry. Existing work has emphasized necessity entry, startup costs, and the distinction between entry intent and realized employer formation \citep{haltiwanger_who_2013, decker_role_2014, bayard_early-stage_2018, dinlersoz_local_2026}. We provide evidence consistent with new firms arising partly because incumbents redraw their boundaries and create demand for specialized suppliers. This helps interpret the rise in U.S.\ business applications alongside falling conversion to employer firms. A supplier-entry channel predicts more applications and more micro-establishments without requiring an improvement in startup quality or a surge of high-growth firms.

In Section~\ref{section:theory} below, we explain the interface-compression mechanism, formalized in Appendix~\ref{appendix:theory_formal}. Section~\ref{section:data} describes our empirical framework and the data. Section~\ref{section:results} presents the main results. Section~\ref{section:conclusion} concludes.

\section{Interface Compression Mechanism} \label{section:theory}

We introduce the intuition behind interface compression and the testable predictions. Appendix~\ref{appendix:theory_formal} formalizes the structuredness cutoff for outsourcing and the conditional mapping from boundary shifts to local specialist entry in a reduced-form model. The model is an organizing device for the empirical comparative static. It is not a full theory of firm boundaries, automation, or entrepreneurial selection.

Consider an auxiliary bundle \(m\), such as medical coding, staffing, facilities support, on-site administrative work, logistics coordination, or local compliance. When the bundle is kept in-house, the firm must manage the governance tasks attached to those services. These include human resources activities, quality monitoring, access coordination, compliance tracking, and exception resolution. Outsourcing does not make these contingencies disappear. It can move much of that governance problem to a supplier and leave the buyer with a thinner interface, meaning a smaller set of verifiable outcomes that governs the business relationship.

Let \(s\) denote how structured and codifiable the bundle is in the buyer's workflow, with higher \(s\) corresponding to more standardized handoffs. Let \(K_m(\tau)\) denote the cost of governing the bundle internally. Let \(p_m+G_m(s,\tau)\) denote the cost of buying it from a specialized supplier, where \(p_m\) is the supplier price, \(G_m(s,\tau)\) is the buyer's interface cost, and \(\tau\) indexes the buyer's governance environment, including monitoring, record-keeping, communication, and procurement technologies. The buyer-side surplus from outsourcing is
\[
S_m(s,\tau) \equiv K_m(\tau)-\big[p_m+G_m(s,\tau)\big].
\]
In this notation, we suppress any direct dependence of internal governance costs on structuredness and focus on the buyer's external interface cost. The substantive requirement is a single-crossing condition. Structuredness must raise outsourcing surplus relative to internal provision.\footnote{More generally, one could write
\[
S_m(s,\tau)=K_m(s,\tau)-[p_m+G_m(s,\tau)].
\]
The required condition is
\[
\frac{\partial S_m(s,\tau)}{\partial s}
=
K_{m,s}(s,\tau)-G_{m,s}(s,\tau)>0.
\]
Thus, if structuredness also lowers internal governance costs, the mechanism requires that it lower the external interface cost more than it lowers the internal governance cost.}
Under this condition, more structured bundles are easier to govern through deliverables, and each bundle has a structuredness cutoff \(s_m^*(\tau)\) above which outside procurement is optimal.

Interface compression refers to the case in which the buyer-facing cost of external governance falls relative to the cost of internal governance. Improvements in the governance environment can lower the outsourcing cutoff, but our empirical analysis does not directly estimate such changes. Instead, it tests whether locations with more initially routine, structured work exhibit the supplier-entry patterns implied by this margin.

This logic does not require entrants to possess a technological advantage in the underlying task. Their advantage may be organizational. They absorb governance-intensive bundles into a vendor relationship, allowing the buyer to manage deliverables rather than the full internal bundle of supervision, coordination, and compliance. When codification and monitoring make more performance dimensions verifiable, the relative value of internal authority may fall. This comparative static is in the spirit of incomplete-contract and property-rights theories of firm boundaries \citep{grossman_costs_1986}. The same logic can also be read through improved verifiability. Record-keeping, software, and monitoring systems may expand the amount of observable information about a bundle. The buyer need not use all of that information directly. The relevant interface may instead be compressed into a smaller set of contractible summary metrics, making arm's-length procurement more attractive \citep{baker_contractibility_2004}.

If this mechanism is operational, locations with a larger share of routine, structured workflows should generate more demand for specialist suppliers. The local-entry prediction is conditional on the outsourced bundle requiring local delivery. Let \(\lambda_m\) denote the local-delivery intensity of bundle \(m\). If \(\lambda_m=0\), outsourcing may occur, but the model has no prediction for local supplier entry because the work can be supplied remotely or from abroad. If \(\lambda_m>0\), some outsourced demand must be served near the buyer, workers, facilities, or customers. In that case, supplier entry should appear in the relevant local service market, which we proxy empirically with commuting zones.

Empirically, we expect this mechanism to be strongest for routine-cognitive auxiliary bundles, such as administrative, compliance, staffing, coordination, and back-office services. Routine-manual work may instead adjust through automation, plant-level restructuring, or nonlocal production sourcing, and therefore need not generate the same local supplier-entry response.

Because the entrant's advantage is organizational rather than frontier technological, the most direct empirical footprint is on the extensive margin. We expect more business applications and more supplier establishments, especially micro-establishments where local service provision is fragmented. Business applications are observed before employer realization, so we interpret them as an early measure of expected supplier-entry demand. Realized establishment births are a downstream validation of the same margin. The mechanism does not require entrants to be higher-quality startups.

The same logic also implies organizational reallocation within incumbent firms. As peripheral bundles move outside firm boundaries, incumbent firms become narrower even as production becomes broader across industries. In observable terms, the mechanism should leave two coarse signatures. First, local economies should have more active narrow supplier industries. Second, incumbent firms should report less secondary production. These are the predictions we take to the data.\footnote{Interface compression may also lower the cost of exploiting inter-firm comparative advantage, so boundary redrawing can generate productivity gains through finer specialization. In our empirical analysis, we treat that as a downstream implication of the mechanism rather than as a separately identified channel.}

%%%%%%%%%%%%%%%%%%%%%%%%%%%%%%%%%%%%%%%%%%%%%%%%%%%%%%%%%%%%%%
\section{Data and Empirical Framework} \label{section:data}
%%%%%%%%%%%%%%%%%%%%%%%%%%%%%%%%%%%%%%%%%%%%%%%%%%%%%%%%%%%%%%

We study whether business entry evolved differently across commuting zones with different initial task composition. Our empirical design relates long-difference changes in entrepreneurship to baseline routine employment share, a task-composition proxy for the local prevalence of structured, standardized work near the codification margin. We describe the construction of \(RSH\) in Section~\ref{section:rsh}, the entrepreneurship outcomes in Section~\ref{section:yvars}, and the identification strategy in Section~\ref{section:estimation}.

\subsection{Measuring Baseline Routine Employment Share}\label{section:rsh}

Because no direct commuting-zone measure of codifiability is available, we use \emph{Routine Employment Share} (\(RSH\)) as a task-composition proxy for the local prevalence of structured, repetitive work. Specifically, \(RSH\) is the share of local employment in occupations whose routine-task intensity, as measured by \citet{deming_growing_2017}, lies above the 66th percentile of the national occupation distribution.\footnote{Formally, we define
\[
RSH_{ct} = \frac{\sum_j Emp_{cjt}\,\mathbf{1}\{RoutineIntensity_j > p_{66}\}}{\sum_j Emp_{cjt}},
\]
where \(Emp_{cjt}\) denotes employment in occupation \(j\) in commuting zone \(c\) and year \(t\), \(RoutineIntensity_j\) is the occupation-level routine-intensity measure from \citet{deming_growing_2017}, and \(p_{66}\) is the 66th percentile of the occupation-level routine-intensity distribution. We construct commuting-zone employment by matching occupations at the six-digit SOC level and aggregating from Public Use Microdata Areas (PUMAs) to commuting zones using Census crosswalks and population allocation factors.}

Because the Deming routine-intensity index averages two O*NET items, the degree of job automation and the importance of repeatedly performing the same physical or mental activities, higher \(RSH\) indicates that a larger share of local employment is concentrated in occupations with repetitive and standardized tasks. We use this measure as a proxy for local exposure to codifiable work near the interface-compression margin. It should not be interpreted as a measure of realized technology adoption or later adoption of any particular tool, process, or equipment.

Appendix Table~\ref{tab:aes_correlations} situates \(RSH\) relative to related task-based measures. \(RSH\) is positively correlated with routine cognitive task content and repeated work measures, and negatively correlated with routine manual task intensity.\footnote{For example, \(RSH\) correlates 0.708 with routine cognitive tasks, 0.854 with the O*NET repetition measure, and \(-0.294\) with routine manual tasks.} These correlations clarify the margin captured by \(RSH\). The measure is especially useful for this paper because it captures structured work that may be easier to govern through deliverables and vendor relationships.

Throughout the empirical analysis, we treat \(RSH\) as a task-based measure of initial routine intensity. Our main specifications use \(RSH_{c,2005}\), measured at the start of the sample window. Because the occupation-level routine scores are national, cross-sectional variation in \(RSH\) comes entirely from differences in local occupational composition. Table~\ref{tab:summary_stats} shows that the mean of \(RSH_{c,2005}\) is 31.4 percentage points, with a standard deviation of 3.1. Appendix Figure~\ref{fig:aes_map} shows substantial geographic dispersion across commuting zones.

\subsection{Entrepreneurship Outcomes}\label{section:yvars}

Our primary outcome is business entry intent, measured using business applications from the Census Bureau's Business Formation Statistics (BFS).\footnote{\url{https://www.census.gov/econ/bfs/}} BFS records Employer Identification Number (EIN) applications, which are typically filed before employer operations begin. Because not all applications become operating employer firms, we interpret BFS as a forward-looking measure of entry intent rather than realized firm formation or firm survival. \citet{bayard_early-stage_2018} show that BFS applications strongly predict subsequent employer business formation.

BFS is available at the county-year level. We aggregate counties to commuting zones using the county-to-commuting-zone crosswalk and scale applications by population to obtain business applications per 100{,}000 residents.\footnote{The crosswalk is from \citet{autor_growth_2013}, available at \url{https://www.ddorn.net/data.htm}. We use Census 2000 population counts aggregated to commuting zones for per-capita scaling.} We focus on the period from 2005 to 2019, when BFS is available before the COVID-19 shock. Figure~\ref{fig:BA_over_time} plots the population-weighted annual mean of business applications per 100{,}000 residents across commuting zones. Applications decline through 2009 and then rise steadily, especially after 2013.

Our main BFS outcomes are long differences in business applications per 100{,}000 residents. Table~\ref{tab:summary_stats} reports the descriptive statistics. The mean change from 2005 to 2019 is 145.8. A decline of 52.7 between 2005 and 2010 is followed by an increase of 198.5 between 2010 and 2019, consistent with weak entry intent during the Great Recession and stronger subsequent growth. Appendix Figure~\ref{fig:busapps_change_map} shows substantial geographic dispersion in the 2005 to 2019 change, with a standard deviation of 599.1.

We also use the Census Bureau's Business Dynamics Statistics (BDS) to examine realized establishment entry. Unlike BFS, which captures applications, BDS records realized employer business dynamics. We construct total establishment entry per 100{,}000 residents at the commuting-zone level using the same crosswalk and population scaling. BDS serves two purposes in the paper. First, it allows us to assess whether the BFS application response maps into realized entry. Second, because it extends well before 2005, it provides the closest available outcome for timing and placebo exercises.

We augment our analysis with County Business Patterns (CBP) to characterize the establishment-level and organizational margins of adjustment. CBP provides annual county-level counts of employer establishments by detailed NAICS industry and employment-size class. We aggregate counties to commuting zones using the same county-to-commuting-zone crosswalk. These data allow us to construct the establishment-size distribution and the number of active four-digit NAICS industries in each commuting zone. Combining CBP employment weights with off-diagonal industry shares from the 2007 BEA Make Table allows us to construct a commuting-zone proxy for exposure to industries that produce outside their primary commodity classification.

Finally, we use data from the Startup Cartography Project \citep{guzman_state_2020}.\footnote{Data are available at \url{https://www.startupcartography.com/data} and via the Harvard Dataverse at \url{https://dataverse.harvard.edu/dataset.xhtml?persistentId=doi:10.7910/DVN/BMRPVH}.} Available through 2016, these data allow us to examine quality-adjusted entrepreneurship beyond entry intent alone. In particular, the startup formation rate (SFR) measures the quantity of new business registrants within a population, the Entrepreneurial Quality Index (EQI) measures the average growth potential within a group of startups, and realized entrepreneurial capital per input (RECPI) measures the quality-adjusted quantity of entrepreneurship, or equivalently the expected number of startups in a region or cohort that later achieve a significant growth outcome.\footnote{In the Startup Cartography Project, RECPI is defined as \(SFR \times EQI\).} Together, BFS, BDS, CBP, and the Startup Cartography outcomes let us distinguish entry intent, realized entry, establishment-scale adjustment, organizational reallocation, and startup quality.

\subsection{Estimation and Instrumental Variables}\label{section:estimation}

Our empirical specification relates long-difference changes in entrepreneurship to baseline routine employment share:
\begin{equation}
\Delta y_c = \beta\,RSH_{c,2005} + X_c'\gamma + \delta_{s(c)} + \varepsilon_c,
\label{eq:ld_baseline}
\end{equation}
where \(\Delta y_c\) is the change in the outcome of interest in commuting zone \(c\) over the relevant window. For commuting zones that cross state borders, \(s(c)\) denotes the state containing the largest share of the commuting zone's 2000 population. We consider business applications per 100{,}000 residents as our primary outcome. \(RSH_{c,2005}\) is baseline routine employment share, \(X_c\) is a vector of controls, and \(\delta_{s(c)}\) are state fixed effects. This control set follows the specification in \citet{autor_growth_2013}.\footnote{\label{fn:controls}
They include the unemployment rate, manufacturing employment share, female manufacturing employment share, the college-to-non-college employment ratio, foreign-born population share, and the ratio of workers aged 21 to 55 to workers aged 56 and older. Census-based controls are constructed using the 2000 Decennial Census microdata via IPUMS USA. The unemployment rate is measured in 2005.} All regressions are weighted by commuting-zone population share, measured using 2000 population. For baseline specifications, we report standard errors clustered by state. The coefficient \(\beta\) measures whether business entry changed more in commuting zones with higher baseline routine employment share.

A central concern is that \(RSH_{c,2005}\) may be measured with error and may also be correlated with persistent local characteristics that predict later entrepreneurship. To address these concerns, we instrument \(RSH_{c,2005}\) with a predetermined shift-share measure based on 1990 industry composition:
\begin{equation}
Z_{c,1990} = \sum_i E_{i,c,1990} \times R_{i,-s(c),1990},
\label{eq:rsh_iv}
\end{equation}
where \(E_{i,c,1990}\) is commuting zone \(c\)'s employment share in industry \(i\) in 1990 and \(R_{i,-s(c),1990}\) is the leave-one-state-out routine employment share of industry \(i\) in 1990. Intuitively, the instrument assigns greater predicted routine intensity to commuting zones that were more specialized in 1990 in industries that employed more routine workers elsewhere in the country. If historical industrial structure predicts baseline occupational routine intensity in 2005, then \(Z_{c,1990}\) should strongly predict \(RSH_{c,2005}\), which we confirm in the first stage.

This empirical strategy is an exposure design based on predetermined local industry composition. The identifying assumption is that, conditional on controls and state fixed effects, 1990 exposure to nationally routine-intensive industries affects post-2005 entrepreneurship through baseline routine occupational structure rather than through other persistent local industry shocks. This is a demanding restriction because historical industry composition may also proxy for durable differences in local demand, supplier networks, human capital, industrial decline, or business formation capacity. We address this concern in three ways. First, the exposure shares are measured in 1990, well before the 2005 to 2019 outcome window. Second, the industry routine shares are computed leave-one-state-out, which reduces mechanical correlation with local shocks and state-specific forces. Third, the specifications absorb broad initial commuting-zone conditions and common state-level changes through controls and state fixed effects. We further assess the restriction using pre-period outcomes, alternative controls, timing exercises, and shift-share diagnostics \citep{adao_shift-share_2019, borusyak_quasi-experimental_2022}.

%%%%%%%%%%%%%%%%%%%%%%%%%%%%%%%%%%%%%%%%%%%%%%%%%%%%%%%%%%%%%%
\section{Results} \label{section:results}
%%%%%%%%%%%%%%%%%%%%%%%%%%%%%%%%%%%%%%%%%%%%%%%%%%%%%%%%%%%%%%

\subsection{Baseline Estimates}

Table~\ref{tab_baseline_results} reports baseline 2SLS estimates for long differences in business applications per 100{,}000 residents. All specifications include state fixed effects and baseline controls. They are weighted by commuting-zone population share, with standard errors clustered by state. Column 1 reports the first stage. The coefficient on the instrument is 0.943, and the Kleibergen--Paap \(rk\) Wald \(F\)-statistic is 81.8. Historical industry composition therefore strongly predicts baseline routine employment share.

The second stage estimates are positive in all three outcome windows. For 2005 to 2019, a 1 percentage-point higher \(RSH_{c,2005}\) is associated with 27.8 additional business applications per 100{,}000 residents. Since the cross-CZ standard deviation of \(RSH_{c,2005}\) is 3.1 percentage points, this corresponds to roughly 86 additional applications per 100{,}000 residents for a one-standard-deviation difference in baseline routine intensity. The estimates are also positive within both subperiods. A 1 percentage-point higher \(RSH_{c,2005}\) is associated with 8.9 additional applications per 100{,}000 residents from 2005 to 2010 and 19.0 additional applications per 100{,}000 residents from 2010 to 2019. The larger post-2010 estimate suggests that the entry response builds over time rather than appearing only in the early part of the sample.

Figure~\ref{fig_year_by_year_BA_long_dif} examines this timing directly. Each point comes from a separate 2SLS regression in which the outcome is the cumulative change in business applications between 2005 and year \(t\). The coefficients are small and imprecise early in the sample. They become larger over time, especially after the mid-2010s. This pattern is consistent with a gradual supplier-entry response in places with more structured, routine work at baseline. It is less consistent with an immediate level shift in business applications.

\subsection{Robustness Checks}

The baseline relationship between \(RSH_{c,2005}\) and business applications is stable across alternative constructions and specifications. Appendix Table~\ref{tab:robustness} reports the main checks. Replacing the 1990 shift-share instrument with a 1980-based version gives a full-period estimate of 34.2 (Panel~A). Estimating the relationship by OLS gives a smaller but still positive and statistically significant coefficient of 15.0 (Panel~B). The estimate also remains positive when baseline controls enter as quartile indicators rather than continuous covariates (Panel~C) and when the outcome is expressed in logs (Panel~D). These checks indicate that the full-period relationship does not hinge on the historical base year, the estimator, or the functional form.

A separate concern is that \(RSH_{c,2005}\) may capture differential exposure to the Great Recession rather than baseline routine intensity. Panel~E addresses this concern by adding a Bartik-style recession control, constructed from 2005 industry shares and leave-one-out national employment growth from 2007 to 2010. The full-period coefficient remains close to the baseline estimate at 28.4. Panel~F redefines \(RSH\) using the 50th, 75th, and 90th percentiles of the national routine-intensity distribution rather than the 66th percentile. The coefficient remains positive in every case. Across these variants, the 2010 to 2019 estimate generally exceeds the 2005 to 2010 estimate, which matches the gradual timing pattern in Figure~\ref{fig_year_by_year_BA_long_dif}.\footnote{A related concern is that the 1990 shift-share instrument may proxy for differential exposure to later state policy changes that independently affect business entry. Appendix Table~\ref{table_policy_balance} addresses this concern by relating the instrument to 2005 to 2019 changes in right-to-work status, corporate tax rates, effective minimum wages, and non-compete enforceability. Across all four outcomes, the estimated relationships are small and statistically indistinguishable from zero.}

Panel~G examines inference. The full-period estimate remains statistically significant under White heteroskedasticity-robust standard errors, census-division clustered standard errors with wild-bootstrap \(p\)-values, and Conley spatial standard errors. Figure~\ref{fig:leave_one_out} adds a state influence diagnostic by re-estimating the baseline model while omitting one state at a time. The resulting coefficients are uniformly positive and remain close to the full-sample benchmark. No single state drives the baseline pattern.

We also examine the industry sources of identifying variation. Following \citet{goldsmith-pinkham_bartik_2020}, Appendix Table~\ref{tab:rotemberg_top10} reports Rotemberg weights using the common-shock approximation. This approximation is nearly identical to the exact leave-one-state-out instrument in the CZ data, with a correlation of 0.9994. The top 10 industries account for slightly more than 50 percent of the total absolute Rotemberg mass. This finding indicates that the baseline estimate is not mechanically driven by one sector.

Appendix Table~\ref{tab:bartik_exact_sensitivity} tests the concern of whether top industries drive the estimates. When the instrument is reconstructed using only the top-weight industries, the full-period estimate remains positive and statistically significant, ranging from 30.4 to 35.8 across the top-1 through top-10 constructions. When the top industries are removed sequentially, the estimate remains positive and of similar magnitude through the first six exclusions, with coefficients between 21.5 and 24.6. The estimates become less precise only after larger shares of the identifying variation are deliberately removed. These results show that the baseline relationship is not driven by a single dominant industry, although a relatively small set of industries accounts for much of the shift-share signal.

Appendix Table~\ref{tab:bartik_pooled_iv} provides a final check by entering the high-weight industries jointly as separate instruments. Using the top five standardized industry shares yields 2SLS and LIML estimates of 26.9 and 27.1. Using the top ten yields estimates of 28.0 and 28.5. These estimates are close to the baseline estimate of 27.8, and the corresponding first stages remain strong.

\subsection{Intended versus Realized Entry}

We next ask whether the increase in business applications appears in realized establishment entry. Table~\ref{tab_BDS} reports 2SLS estimates using the change in BDS establishment entries per 100{,}000 residents as the outcome. The coefficients are 2.87 for 2005 to 2019, 0.27 for 2005 to 2010, and 2.60 for 2010 to 2019. Because average BDS establishment entry declines over the full period in our sample, these estimates should be read comparatively. Commuting zones with higher baseline \(RSH_{c,2005}\) experienced a smaller decline, or faster growth, in realized establishment entry than otherwise similar commuting zones with lower baseline routine employment share. The realized-entry signal is positive, and most of it comes after 2010.

The realized-entry response is smaller than the business-application response. Over the full period, the coefficient in Table~\ref{tab_BDS} is roughly one-tenth of the baseline BFS coefficient in Table~\ref{tab_baseline_results}. That difference is expected. BFS captures entry intent, while BDS captures the narrower set of businesses that become operating employer establishments. At the same time, the positive BDS estimates make it harder to read the BFS results as filing behavior or administrative churn alone.

BDS data also provide us with a timing window that BFS cannot, since the BDS series extends well before 2005. Figure~\ref{fig:bds_timing_annual_2sls} plots annual 2SLS estimates relating baseline \(RSH_{c,2005}\) to cumulative changes in BDS establishment entry relative to 2005.\footnote{For years before 2005, the points are backward cumulative differences relative to 2005. For years after 2005, the points are cumulative changes since 2005.} The pre-2005 coefficients show no stable positive buildup. If anything, they are mostly negative and drift toward zero as 2005 approaches. Positive coefficients emerge only later in the post-2005 period, especially after 2015. This pattern is difficult to reconcile with a story in which high-\(RSH\) commuting zones were already on a stronger entry path before 2005. Appendix Figures~\ref{fig:bds_timing_annual_rf} and \ref{fig:bds_timing_5yr_2sls} reinforce this interpretation. Neither the reduced-form annual estimates nor the non-overlapping five-year windows reveal a stable positive pre-2005 relationship.

Table~\ref{tab_size_shares} provides evidence on where the realized response appears in the establishment-size distribution. Using County Business Patterns data, we find that higher \(RSH_{c,2005}\) shifts activity toward micro-establishments and away from establishments with 5 to 19 employees. Over 2005 to 2019, a 1 percentage-point increase in \(RSH_{c,2005}\) raises the share of establishments with 1 to 4 employees by 0.197 percentage points and lowers the share with 5 to 19 employees by 0.111 percentage points. This pattern is sharper from 2010 to 2019. This finding is consistent with interface compression. Boundary redrawing should generate more entries on the extensive margin, especially among small local suppliers. Following \citet{haltiwanger_who_2013}, we interpret this as evidence on the scale and composition of entry rather than as support for a generic small-firm growth.

\subsection{Quality of Entrants}

The mechanism predicts more supplier entry, not necessarily better startups. We therefore ask whether the increase in business applications is accompanied by higher quality entrepreneurship. To do so, we use the Startup Cartography data. Because these data are available only through 2016, all outcomes are measured as long differences from 2005 to 2016. The sample is restricted to the 680 commuting zones covered by Cartography.

Column~(1) of Table~\ref{tab:cartography} re-estimates the BFS application specification on this matched sample. The coefficient remains positive and precisely estimated, although it is smaller than the full-period estimate in Table~\ref{tab_baseline_results}. The application response is therefore still visible in the sample used to study startup quality.

Columns~(2) to (4) replace application counts with Startup Cartography outcomes, which separate startup formation from startup quality.\footnote{The Startup Formation Rate (SFR) measures realized startup formation. The Entrepreneurial Quality Index (EQI) measures the predicted probability that a startup reaches a high growth outcome such as an IPO or major acquisition based on characteristics observable at or near registration. RECPI measures realized entrepreneurial capital per input.} The SFR coefficient is positive but imprecisely estimated. The EQI coefficient is negative. The RECPI estimate is small and statistically insignificant.

These estimates point to a quantity margin rather than a quality margin. Higher baseline \(RSH_{c,2005}\) predicts more business applications in the matched sample, but the additional entry does not translate into higher average startup quality or more quality adjusted entrepreneurial capital. The negative EQI point estimate is consistent with small scale supplier formation rather than high growth entrepreneurship, although it is not statistically distinguishable from zero.

\subsection{Routine Task Content and Robot Exposure}

The evidence so far links baseline \(RSH\) to later business applications and micro-establishment entry. A natural concern is that \(RSH_{c,2005}\) may be standing in for later robot penetration rather than the interface-compression mechanism highlighted in Section~\ref{section:theory}. This concern is useful because the two interpretations have different empirical footprints. A robot-centered interpretation should be strongest where routine work is manual and production-based. The interface-compression mechanism should be strongest where routine work is cognitive, codifiable, and easier to package into deliverables, vendor relationships, and local business services.

This distinction matters because \(RSH\) is an ex ante measure of task content. It captures whether local employment was concentrated in repetitive and standardized work at baseline. It is not a measure of realized robot adoption. The same task structure can make work easier to implement through machinery in some settings, but it can also make work easier to separate from an incumbent firm and govern through a thinner buyer-supplier interface. We therefore use two checks. First, we split routine work into cognitive and manual components. Second, we condition directly on instrumented robot exposure.

Table~\ref{tab:horse_race_tasks} implements the first check. It replaces \(RSH\) with two endogenous variables, the commuting-zone employment shares in occupations above the 66th percentile of \citet{acemoglu_skills_2011}'s routine cognitive and routine manual task-intensity measures. Each component is instrumented with its own leave-one-state-out shift-share IV, constructed from 1990 industry shares and industry-level routine cognitive and routine manual employment shares.\footnote{The split Bartik IVs are constructed in parallel with the baseline 1990 leave-one-state-out instrument. For each commuting zone \(j\), we compute
\[
iv^{cog}_j = \sum_i E_{i,j,1990}R^{cog}_{i,-s(j),1990}
\quad \text{and} \quad
iv^{man}_j = \sum_i E_{i,j,1990}R^{man}_{i,-s(j),1990}.
\]
Here, \(E_{i,j,1990}\) is industry \(i\)'s share of total employment in commuting zone \(j\) from the 1990 Census. \(R^{cog}_{i,-s(j),1990}\) and \(R^{man}_{i,-s(j),1990}\) are industry \(i\)'s national shares of employment in occupations above the 66th percentile of the DOT \texttt{task\_routine} and \texttt{task\_manual} distributions, computed with workers in commuting zone \(j\)'s state excluded. Industries differ in their cognitive and manual routine intensity, so 1990 industrial composition provides separate predictors for the two margins.} The split instruments have enough power to distinguish the two margins. The Sanderson and Windmeijer statistics are 28.7 and 27.7, and the Kleibergen and Paap weak-identification statistic is 12.6.

The business-application response loads on the cognitive side of routine work. Over 2005 to 2019, a 1 percentage-point higher routine-cognitive share is associated with 43.1 additional business applications per 100{,}000 residents. The corresponding routine-manual coefficient is \(-5.9\) and statistically indistinguishable from zero. The same pattern appears in the subperiods. From 2005 to 2010, the routine-cognitive coefficient is 26.1 and precisely estimated, while the routine-manual coefficient is 6.0 and insignificant. From 2010 to 2019, the routine-cognitive estimate remains positive at 16.9 but becomes imprecise, while the routine-manual estimate is \(-12.0\) and only marginally significant.\footnote{Appendix Table~\ref{tab:alt_treatments} reports one-at-a-time specifications rather than the horse-race specification. The corresponding estimates are positive for the O*NET degree-of-automation and repetition measures, and for routine cognitive task intensity, but negative for routine manual task intensity.} These estimates point to the routine-cognitive margin as the relevant source of the business-entry response. That is the margin closest to interface compression. Administrative support, compliance, staffing, medical coding, logistics coordination, and back-office services are often repetitive enough to describe through deliverables while still tied to local buyers. Routine-manual exposure does not show the same supplier-entry pattern.

Table~\ref{tab_control_robots_china} implements the second check by re-estimating the baseline 2SLS specification while controlling for robot exposure, instrumented using European robot adoption following \citet{acemoglu_robots_2020}. The estimated relationship with \(RSH_{c,2005}\) barely moves. Over 2005 to 2019, the coefficient falls from 27.8 in the baseline specification to 25.3 after robot exposure is included. The subperiod estimates are also close to the baseline estimates. They are 7.8 from 2005 to 2010 and 17.4 from 2010 to 2019, compared with baseline estimates of 8.9 and 19.0.

Robot exposure itself enters positively and significantly, which suggests that realized robot penetration may have its own relationship with business entry. That relationship is separate from the \(RSH\) margin. The interface-compression channel operates through codifiable work, thinner buyer-supplier interfaces, and demand for specialist suppliers. Conditioning on robot exposure does not materially attenuate the coefficient on \(RSH_{c,2005}\), so the baseline relationship is not simply a proxy for realized robot penetration.\footnote{A separate specification adding the China import shock reaches the same conclusion. When the China shock is included as an additional control, the coefficients on \(RSH_{c,2005}\) are 27.7 over 2005 to 2019, 9.4 over 2005 to 2010, and 18.3 over 2010 to 2019. The China-shock coefficients are small and statistically insignificant throughout.}

Taken together, these results support the mechanism's interpretation of \(RSH\). The entry response is tied to routine-cognitive task structure rather than to routine-manual exposure, and it remains after conditioning on realized robot exposure. The evidence, therefore, points toward codifiable supplier entry and boundary redrawing, not a generic robot-exposure explanation.

\subsection{Empirical Evidence for Interface Compression} \label{section:mechanism}

The interface compression mechanism implies an organizational footprint beyond the entry response documented above. Incumbent firms should narrow their production scope as peripheral bundles move outside firm boundaries, and buyer–supplier relationships should shift toward arm's-length governance. We bring two kinds of evidence, establishment-level evidence on local industry breadth and incumbent production specialization (Table~\ref{tab:combined_io_org}) and contract-language evidence on how buyer–supplier relationships are governed (Table~\ref{tab:contract_projection}).

Table~\ref{tab:combined_io_org} provides establishment-level evidence on the two organizational signatures implied by the mechanism. Panel~A uses the change in the number of active four-digit NAICS industries in a commuting zone as a measure of local industry breadth. An industry is active if County Business Patterns records at least one establishment in it. A positive coefficient indicates that higher-exposure commuting zones host a wider range of narrow industries, consistent with specialist suppliers absorbing activities that incumbents shed. Panel~B uses the change in the commuting-zone secondary-production share, defined as the employment-weighted share of local activity in industries producing outside their primary commodity classification, with industry-level off-diagonal shares drawn from the 2007 BEA Make Table. A negative coefficient indicates that production profiles narrow toward primary activities, consistent with incumbents shedding peripheral bundles. Column~(1) reports the baseline specification using \(RSH_{2005}\). Columns~(2) and (3) decompose this baseline into routine cognitive and routine manual exposure, entered separately in columns~(2) and (3) and jointly in the horse-race specification.

Panel~A shows that a one percentage-point increase in \(RSH_{2005}\) is associated with 5.75 additional active four-digit industries in a commuting zone. In Panel~B, the same increase predicts a decline in the secondary-production share. The decomposition then clarifies the source of those baseline relationships. When routine cognitive exposure enters alone, a 1 percentage-point increase is associated with additional active industries and a decline in secondary production. In the horse-race specification in column (4), the routine-cognitive estimates retain their sign, magnitude, and significance once routine manual exposure is included as a competing regressor. Routine manual exposure, by contrast, carries the opposite sign when entered alone and attenuates to a small, statistically insignificant coefficient in the joint specification.

Table~\ref{tab:contract_projection} turns from organizational outcomes in the local economy to the governance mode of buyer--supplier relationships in the industries to which commuting zones are exposed. Using SEC EDGAR Exhibit~10 filings,  we measure industry-level changes in buyer--supplier contract language between 2004--2006 and 2017--2019, and project those shifts to commuting zones using 2005 industry employment shares. SEC filings do not identify the location of the underlying buyer--supplier relationship, so we do not observe local contracting directly. Instead, we treat publicly traded filers as a window into industry-level contracting practices, and ask whether high-\(RSH_{2005}\) commuting zones are disproportionately exposed to industries whose contracting environment shifted in the direction implied by the mechanism. The maintained assumption is not that all firms within an industry write identical contracts, but that shifts in contract language among public filers are informative about broader changes in that industry's contracting environment. The features we examine, namely performance-and-deliverables language, automation-and-technology language, and vendor-management vocabulary, correspond to the thinner, more codified buyer interface that the mechanism predicts. To distinguish this mechanism-relevant shift from generic drafting-style drift, we also examine credit agreements filed by the same public firms as a placebo. Credit agreements are not used to govern outsourcing relationships, so a contract-language shift driven by boilerplate trends should appear there as well.

Columns~(1) and (2) use the change in the share of B2B contracts (service, supply, and license agreements) containing performance-and-deliverables keywords as the projected industry measure. Columns~(3) and (4) use the  analogous change for automation-and-technology keywords. Columns~(5) and (6) use a $z$-score composite of five  projected industry-level measures, i.e., the performance-and-deliverables shift, the automation-and-technology shift, the credit-minus-B2B word-count gap, the tier-2 vendor-management keyword share, and B2B contract volume per industry employee.\footnote{Keywords are identified via regex matches to the full contract text. The performance-and-deliverables set contains \textit{performance-based}, \textit{statement of work}, \textit{service-level agreement}, \textit{service-level}, \textit{SLA}, \textit{deliverable} and its variants, \textit{outcome-based}, \textit{outcomes-based}, \textit{key performance indicator} and its variants, and \textit{KPI}. The automation-and-technology set contains \textit{cloud-based}, \textit{cloud-computing}, \textit{cloud-service}, \textit{cloud-platform}, \textit{SaaS}, \textit{software-as-a-service}, \textit{machine-learning}, \textit{artificial intelligence}, \textit{data-analytics}, \textit{automated}, \textit{automation}, \textit{automatic}, \textit{algorithm} and its variants, \textit{application programming interface}, \textit{API}, \textit{digital platform}, \textit{data-exchange}, and \textit{data integration}. The tier-2 vendor-management set contains \textit{vendor} and its variants, \textit{service provider} and its variants, \textit{managed service} and its variants, \textit{service level}, and \textit{third-party}. A contract is flagged for a set if any of its keywords appear at least once.} Columns~(7) and (8) use changes in credit-agreement word counts as a placebo, because credit agreements are filed by the same firms but do not govern outsourcing relationships. Higher values of the projected contract measure therefore indicate greater exposure to industries whose contracts shifted in the indicated direction over the sample period. Within each column pair, the first column uses baseline \(RSH_{2005}\), instrumented with the paper's 1990 leave-one-state-out Bartik IV. The second column reports the corresponding horse-race specification, instrumenting routine cognitive and routine manual exposure jointly with their split Bartik IVs.

Baseline \(RSH_{2005}\) is positively associated with exposure to industries whose B2B contracts shifted toward performance-and-deliverables language, automation-and-technology language, and the composite. The three mechanism-relevant outcomes all load in the direction implied by interface compression. The credit-agreement placebo is small and statistically indistinguishable from zero, so high-\(RSH\) commuting zones are therefore not exposed to industries in which contract language changed generically across all document types. The horse-race specifications attribute these shifts to routine cognitive employment, while routine manual exposure is small and imprecisely estimated. Two further pieces of evidence support reading these projected shifts as a reorganization of buyer--supplier governance rather than a coincidental change in legal drafting. Within the same SEC filers, B2B contracts become shorter and less segmented relative to credit agreements over the sample window, with the compression most pronounced for service agreements (Appendix Table~\ref{tab:contract_complexity}). Because both document types come from the same firms, this within-filer gap is difficult to attribute to generic drafting trends. \footnote{Panel~C of Appendix Table~\ref{tab:contract_complexity} shows that the number of distinct B2B counterparties per filer rose while the number of B2B contracts filed per firm stayed flat, consistent with lead firms spreading vendor work across a wider set of specialist suppliers rather than deepening relationships with the same counterparties. Appendix Table~\ref{tab:govt_contract_categories} reports an external benchmark from federal procurement. Performance-based contracting rises more in codifiable service categories such as professional and administrative support, IT and telecommunications, and utilities and housekeeping, and rises less in bespoke categories such as construction and architecture and engineering. Because the federal government is not an SEC filer, this pattern cannot be explained by common drafting practices at public firms.}

The two tables converge on the same reorganization from independent data sources. Table~\ref{tab:combined_io_org} shows that local production becomes broader across industries even as it narrows within incumbent firms. Table~\ref{tab:contract_projection} shows that high-\(RSH\) commuting zones are more exposed to industries whose buyer--supplier contracts shift toward language suited to arm's-length, codified governance. In both tables, the effect is driven by routine cognitive employment rather than routine manual, consistent with the mechanism's emphasis on codifiable work near the interface-compression margin. Simple alternatives to the mechanism are difficult to reconcile with this combination of patterns. A generic local demand boom can raise the number of active industries, but has no reason to predict a narrower incumbent production scope. A uniform decline in entry costs can raise entry and industry breadth but has no reason to predict differential exposure to outsourcing-oriented contract language, nor the null credit-agreement placebo. Neither alternative predicts differential loading on routine cognitive rather than routine manual exposure. Boundary redrawing, in which structured, codifiable bundles move outside incumbent firms and are absorbed by local specialist suppliers, fits the pattern combination.

\section{Summary and Conclusion} \label{section:conclusion}

We document that from 2005 to 2019, business applications rose disproportionately in U.S.\ commuting zones with a higher baseline routine employment share. A one-standard-deviation difference in baseline routine employment share is associated with roughly 86 additional business applications per 100{,}000 residents. Realized establishment entry also moves in the same direction, although the response is smaller than the application response. At the same time, the activity shifts toward micro-establishments, and we find no improvement in startup quality. The relationship is driven by routine cognitive work rather than routine manual.

These patterns fit the interface-compression mechanism developed in the paper. When work is structured, codifiable, and easier to verify through deliverables, some auxiliary bundles become easier to govern outside the firm. The relevant margin is not that the underlying task disappears or becomes technologically superior. It is that the buyer can replace part of the internal governance problem with a thinner interface to a specialist supplier. When those outsourced bundles still require proximity to buyers, workers, facilities, or customers, boundary redrawing can generate local supplier entry.

The organizational evidence points in the same direction. Higher routine employment in a commuting zone is associated with broader local industry presence, lower secondary production, and contract language that moves toward performance, deliverables, and codified buyer governance. These facts are difficult to explain with a generic local demand boom alone. A demand boom could raise entry. It would not naturally predict a simultaneous shift toward narrower production profiles and contract language better suited to external procurement.

The paper, therefore, adds a domestic supplier-entry margin to the literature on routine and codifiable work. Prior research shows that routineness can make activities easier to source at arm's length across borders \citep{costinot_adaptation_2011, oldenski_task_2012, oldenski_offshoring_2014}. We show that this is not the only margin. When organizationally separable activities remain place-bound, the same routine-cognitive structure can create demand for domestic specialist suppliers. Routine work can therefore leave the incumbent firm's payroll without leaving the local economy.

This supplier-entry channel also helps explain an aggregate pattern in U.S.\ business dynamism. Business applications increased sharply over our sample period, while the share of applications that became employer firms fell. Boundary redrawing can produce exactly this footprint. It creates demand for more applications, more small firms, and more local suppliers, but not necessarily for better startups or more high-growth firms. The result is a different view of business formation. Some new firms arise not because entrepreneurs discover new products or frontier technologies, but because incumbent firms redraw their boundaries and create markets for specialized suppliers.

\pagebreak
\clearpage

%==============================================================================
% BIBLIOGRAPHY
%==============================================================================

\bibliographystyle{aea}
\bibliography{references}
\clearpage
\section*{Figures}
\input{figures_duha/BA_over_time}
\input{figures_duha/year_by_year_BA_long_dif}
\clearpage

\input{figures_duha/bds_timing_2sls}
%==============================================================================

\clearpage
\section*{Tables}

\input{tables_duha/tab_summary_stats}
\clearpage

\input{tables_duha/tab_baseline_results}

\input{tables_duha/BDS_entry}

\input{tables_duha/table_cbp_size}

\input{tables_duha/table_cartography}

\input{tables_duha/horse_race}

\input{tables_duha/tab_control_robots_china}
\clearpage

\input{tables_duha/table_combined_io_organizational}
\clearpage

\input{tables_duha/contract_projection}

%%%%%%%%%%%%%%%%%%%%%%%%%%%%%%%%%%%%%%%%%%%%%%%%%%%%%%%%%%%%%%%%%%%%%%%%%%%%%%%
% APPENDIX
%%%%%%%%%%%%%%%%%%%%%%%%%%%%%%%%%%%%%%%%%%%%%%%%%%%%%%%%%%%%%%%%%%%%%%%%%%%%%%%

\appendix

\renewcommand{\thetable}{A\arabic{table}}
\setcounter{table}{0}
\renewcommand{\thefigure}{A\arabic{figure}}
\setcounter{figure}{0}
\FloatBarrier
\section{Appendix} \label{appendix}

\input{figures_duha/DEMING_map}

\input{figures_duha/ba_app_map}

\input{figures_duha/leave_one_out}

\input{figures_duha/bds_timing_rf}

\input{figures_duha/bds_timing_5year_2sls}
\clearpage

\input{tables_duha/tab_aes_correlations}
\newpage

\input{tables_duha/tab_robustness}

\input{tables_duha/policy_balance}

\input{tables_duha/tab_rotemberg}

\input{tables_duha/bartik_sensitivity}

\input{tables_duha/tab_pooled_share_iv}
\clearpage

\input{tables_duha/table_alternative_treatments}
\clearpage

\input{tables_duha/table_contract_complexity.tex}

\input{tables_duha/table_govt_contract_categories.tex}

\clearpage

\section{Formal Model of Interface Compression}
\label{appendix:theory_formal}

This appendix presents a reduced-form model of boundary choice and local specialist entry. The model is an organizing device for the empirical comparative static rather than a self-contained theory of firm boundaries. The goal is not to build a full property-rights model with explicit asset ownership, bargaining, or residual control rights. Instead, the setup isolates two margins emphasized in the paper. First, more structured bundles are easier to govern through arm's-length procurement. Second, when outsourced bundles require local delivery, the resulting supplier demand can generate local specialist entry. We allow buyer-side governance technologies to shift the outsourcing cutoff, but the empirical analysis does not identify this technology margin directly. It uses cross-location variation in initial routine task structure as a proxy for the prevalence of structured bundles.

\subsection{Environment}

A buyer \(i\) requires one unit of each auxiliary bundle \(m\in\{1,\ldots,M\}\) to operate its core business. Bundle \(m\) generates \(n_m\ge 1\) underlying buyer-relevant contingencies, such as scheduling margins, access coordination, compliance requirements, and exception handling. Let \(\tau\) index the buyer-side governance environment, with higher \(\tau\) corresponding to greater standardization, monitoring capacity, record-keeping, communication, and ease of codification.

Buyers differ in how structured a given bundle is in their own workflow. Let \(s_{im}\in[0,1]\) denote the \emph{structuredness} of bundle \(m\) for buyer \(i\), where higher \(s_{im}\) means that the bundle can be governed through more standardized and codifiable handoffs. To distinguish organizational separation from geographic relocation, let \(\lambda_m\in[0,1]\) denote the local-delivery intensity of bundle \(m\). When \(\lambda_m=0\), the bundle can be supplied remotely or from abroad without local presence. When \(\lambda_m=1\), the bundle must be supplied locally because delivery depends on proximity to buyers, workers, facilities, or customers. Intermediate values capture bundles for which only part of outsourced demand is local.

Abstracting away from differences in frontier task productivity by treating the direct production cost of executing the underlying task as common across governance modes and normalizing it out, we can model the boundary decision as depending only on governance costs. By construction, heterogeneity in \(s_{im}\) affects the governability of outsourcing rather than direct task productivity or internal governance cost. Nothing in the logic requires this restriction. The same cutoff argument extends more generally so long as the resulting outsourcing surplus remains increasing in \(s_{im}\). If bundle \(m\) is kept in-house, buyer \(i\) incurs internal governance cost
\[
K_m(n_m,\tau),
\]
with
\[
K_n>0,\qquad K_{nn}\ge 0,\qquad K_\tau\le 0.
\]
Thus, internal governance becomes more costly as contingencies proliferate, though improvements in the governance environment may reduce those costs. If the bundle is procured from an outside specialist, the buyer pays supplier price \(p_m\) and incurs a buyer-side interface cost.

Define
\[
d_{im}\equiv \psi(n_m,s_{im},\tau),
\]
where \(d_{im}\) is \emph{interface dimensionality}, i.e., the size of the minimal buyer-facing verifiable performance vector sufficient to govern bundle \(m\) at arm's length. We assume
\[
\psi_n>0,\qquad \psi_s<0,\qquad \psi_\tau<0,\qquad \psi(n,s,\tau)\le n.
\]
The last inequality states that the buyer-facing interface is a compressed representation of the underlying contingency set. The object \(d_{im}\) is not the total amount of information generated by the task. Monitoring, record-keeping, and communication systems may increase the amount of raw observable information while still reducing \(d_{im}\) by compressing that information into a smaller set of contractible summary metrics.

Let
\[
C(d_{im},\tau)
\]
denote the buyer's cost of specifying, monitoring, and administering the outsourcing interface, with
\[
C_d>0,\qquad C_\tau<0.
\]
From the buyer's perspective, \(p_m\) is taken as given. In the entry block below, it is pinned down by supplier operating costs under competitive entry. Allowing \(p_m\) to vary with \(\tau\) would add a supply-side channel, but is not needed for the buyer-side mechanism emphasized here.

\subsection{Buyer's Boundary Choice}

For each bundle \(m\), buyer \(i\) chooses whether to keep the bundle in-house or outsource it. Let
\[
x_{im}\in\{0,1\},
\]
where \(x_{im}=1\) denotes internal provision and \(x_{im}=0\) denotes outsourcing. The buyer solves
\begin{equation}
\min_{\{x_{im}\}_{m=1}^M}
\sum_{m=1}^M
\left[
x_{im}K_m(n_m,\tau)
+
(1-x_{im})\Big(p_m + C(\psi(n_m,s_{im},\tau),\tau)\Big)
\right].
\label{eq:buyer_problem}
\end{equation}
Because the objective is separable across bundles, the boundary choice is bundle by bundle. Relative to the notation in the main text, the buyer-side interface cost is
\[
G_{im}(\tau)\equiv C(\psi(n_m,s_{im},\tau),\tau).
\]

Define outsourcing surplus as
\begin{equation}
S_{im}(s_{im},\tau)
\equiv
K_m(n_m,\tau)
-
\Big[
p_m + C(\psi(n_m,s_{im},\tau),\tau)
\Big].
\label{eq:outsourcing_surplus}
\end{equation}
Buyer \(i\) outsources bundle \(m\) if and only if
\begin{equation}
S_{im}(s_{im},\tau)>0.
\label{eq:outsourcing_rule}
\end{equation}
Equivalently,
\[
x_{im}^*(s_{im},\tau)
=
\mathbf{1}\!\left\{
K_m(n_m,\tau)
\le
p_m + C(\psi(n_m,s_{im},\tau),\tau)
\right\}.
\]

For a generic buyer with structuredness level \(s\), define
\begin{equation}
S_m(s,\tau)
\equiv
K_m(n_m,\tau)
-
\Big[
p_m + C(\psi(n_m,s,\tau),\tau)
\Big].
\label{eq:generic_outsourcing_surplus}
\end{equation}
The buyer-specific surplus in \eqref{eq:outsourcing_surplus} is therefore \(S_{im}(s_{im},\tau)=S_m(s_{im},\tau)\).

\subsection{Comparative Statics in Structuredness and Buyer-Side Governance}

Differentiating \eqref{eq:generic_outsourcing_surplus} yields
\begin{equation}
\frac{\partial S_m(s,\tau)}{\partial s}
=
-
C_d(\psi(n_m,s,\tau),\tau)\,
\psi_s(n_m,s,\tau)
>0,
\label{eq:Ss}
\end{equation}
since \(C_d>0\) and \(\psi_s<0\). That is, more structured bundles are more likely to be outsourced.

Similarly,
\begin{equation}
\frac{\partial S_m(s,\tau)}{\partial \tau}
=
K_\tau(n_m,\tau)
-
C_d(\psi(n_m,s,\tau),\tau)\,
\psi_\tau(n_m,s,\tau)
-
C_\tau(\psi(n_m,s,\tau),\tau).
\label{eq:Stau}
\end{equation}
Improvements in the buyer-side governance environment may lower both internal and external governance costs. Interface compression refers to the case in which the external margin improves more, so that
\begin{equation}
-
C_d(\psi(n_m,s,\tau),\tau)\,
\psi_\tau(n_m,s,\tau)
-
C_\tau(\psi(n_m,s,\tau),\tau)
>
-
K_\tau(n_m,\tau).
\label{eq:relative_improvement_appendix}
\end{equation}
Under \eqref{eq:relative_improvement_appendix}, outsourcing surplus is increasing in \(\tau\). This comparative static is included to formalize the interface-compression logic. The empirical analysis, however, primarily exploits cross-location variation in the initial distribution of structured work rather than direct variation in \(\tau\).

\begin{proposition}[Structuredness cutoff]
\label{prop:structuredness_cutoff}
Fix bundle \(m\). Suppose \(S_m(s,\tau)\) is continuously differentiable in \((s,\tau)\), strictly increasing in \(s\), and satisfies
\[
S_m(0,\tau)<0<S_m(1,\tau)
\]
for the relevant range of \(\tau\). Then there exists a unique cutoff \(s_m^*(\tau)\in(0,1)\) such that buyer \(i\) outsources bundle \(m\) if and only if
\[
s_{im}> s_m^*(\tau).
\]
If, in addition, \eqref{eq:relative_improvement_appendix} holds, then \(s_m^*(\tau)\) is decreasing in \(\tau\), with
\[
\frac{d s_m^*(\tau)}{d\tau}
=
-
\frac{
\partial S_m(s_m^*(\tau),\tau)/\partial \tau
}{
\partial S_m(s_m^*(\tau),\tau)/\partial s
}
<0.
\]
\end{proposition}

\begin{proof}
Equation \eqref{eq:Ss} implies that \(S_m(s,\tau)\) is strictly increasing in \(s\). Continuity and the sign change between \(s=0\) and \(s=1\) imply the existence of a unique interior cutoff \(s_m^*(\tau)\) satisfying
\[
S_m(s_m^*(\tau),\tau)=0.
\]
The outsourcing rule \eqref{eq:outsourcing_rule} then implies that buyer \(i\) outsources bundle \(m\) if and only if \(s_{im}> s_m^*(\tau)\). At the cutoff itself, the buyer is indifferent. Under the convention in \eqref{eq:outsourcing_rule}, equality does not count as outsourcing. Under \eqref{eq:relative_improvement_appendix}, equation \eqref{eq:Stau} implies \(\partial S_m(s,\tau)/\partial \tau>0\). Applying the implicit function theorem to \(S_m(s_m^*(\tau),\tau)=0\) yields
\[
\frac{d s_m^*(\tau)}{d\tau}
=
-
\frac{
\partial S_m(s_m^*(\tau),\tau)/\partial \tau
}{
\partial S_m(s_m^*(\tau),\tau)/\partial s
}.
\]
Since both numerator and denominator are positive, the derivative is negative.
\end{proof}

The sign-change condition rules out corner cases in which a bundle is always internal or always outsourced over the relevant range. Those corner cases are allowed in the broader framework, but the interior cutoff is the empirically relevant case for bundles near the make-or-buy margin.

\subsection{Specialist Entry}

To connect the boundary choice to entrepreneurship, consider a location \(\ell\) and normalize the mass of potential buyers requiring bundle \(m\) to one. With a larger buyer base, the expressions below scale proportionally. Let \(H_{\ell m}(s)\) denote the cumulative distribution function of structuredness across those buyers. The mass of buyers that outsource bundle \(m\) in location \(\ell\) is
\begin{equation}
D_{m\ell}(\tau)
=
\int \mathbf{1}\{s>s_m^*(\tau)\}\, dH_{\ell m}(s)
=
1-H_{\ell m}(s_m^*(\tau)).
\label{eq:outsourced_demand}
\end{equation}

This expression captures outsourced demand generated by buyers in location \(\ell\). It does not by itself imply that the supplier is also located in \(\ell\). Local supplier demand depends on the bundle's local-delivery intensity \(\lambda_m\). Define local outsourced demand as
\begin{equation}
D^{L}_{m\ell}(\tau)
=
\lambda_m D_{m\ell}(\tau)
=
\lambda_m\left[1-H_{\ell m}(s_m^*(\tau))\right].
\label{eq:local_outsourced_demand}
\end{equation}
When \(\lambda_m=0\), outsourcing may be supplied remotely or from abroad and generates no local entry prediction. When \(\lambda_m>0\), part of the boundary shift creates demand for nearby suppliers.

Now introduce a minimal local supply side. A specialist establishment serving bundle \(m\) incurs fixed operating cost \(F_m\), constant net operating cost \(c_m\) per client, and can serve \(q_m>0\) local buyers. Under competitive entry and full utilization,
\[
p_m = c_m + \frac{F_m}{q_m},
\]
so the supplier price in the buyer problem covers operating costs. Given local outsourced demand \(D^{L}_{m\ell}(\tau)\), the equilibrium mass of specialist establishments in location \(\ell\) is
\begin{equation}
E_{m\ell}(\tau)=\frac{D^{L}_{m\ell}(\tau)}{q_m}.
\label{eq:entry}
\end{equation}

\begin{proposition}[Local entry response]
\label{prop:entry}
If location \(\ell\) has a first-order stochastically higher distribution of structuredness than location \(\ell'\), then
\[
D^{L}_{m\ell}(\tau)\ge D^{L}_{m\ell'}(\tau)
\qquad\text{and}\qquad
E_{m\ell}(\tau)\ge E_{m\ell'}(\tau)
\]
for every \(\tau\), provided the two locations have the same \(\lambda_m\) and \(q_m\). If, in addition, \(s_m^*(\tau)\) is decreasing in \(\tau\) and \(\lambda_m>0\), then local outsourced demand \(D^{L}_{m\ell}(\tau)\) and local specialist entry \(E_{m\ell}(\tau)\) are weakly increasing in \(\tau\).
\[
D^{L}_{m\ell}(\tau)\ge D^{L}_{m\ell'}(\tau)
\qquad\text{and}\qquad
E_{m\ell}(\tau)\ge E_{m\ell'}(\tau)
\]
for every \(\tau\), provided the two locations have the same \(\lambda_m\) and \(q_m\).
\end{proposition}

\begin{proof}
Equation \eqref{eq:local_outsourced_demand} shows that local demand is the outsourced upper tail of the structuredness distribution multiplied by the local-delivery intensity \(\lambda_m\). When the cutoff falls, the upper tail weakly expands, so \(D^{L}_{m\ell}(\tau)\) weakly rises whenever \(\lambda_m>0\). Equation \eqref{eq:entry} then implies the same result for \(E_{m\ell}(\tau)\). The cross-location comparison follows from first-order stochastic dominance.
\end{proof}

This entry block is intentionally minimal. Its purpose is not to explain which individual entrepreneur enters, but to show how a buyer-side boundary shift can create market demand for local specialist suppliers when outsourced bundles remain place-bound. In the empirical analysis, \(RSH\) is treated as a reduced-form location-level proxy for places in which more buyer-bundle pairs lie near, and on the outsourceable side of, this structuredness margin.

\subsection{Implications for Incumbent Scope and Industry Reallocation}

Let buyer \(i\)'s within-firm auxiliary scope be
\begin{equation}
W_i(\tau)=\sum_{m=1}^M x_{im}^*(s_{im},\tau),
\label{eq:scope}
\end{equation}
the number of auxiliary bundles still performed inside the firm. Since \(x_{im}^*(s_{im},\tau)\) switches from one to zero when the falling cutoff \(s_m^*(\tau)\) passes below the fixed structuredness level \(s_{im}\), Proposition~\ref{prop:structuredness_cutoff} implies that \(W_i(\tau)\) is weakly decreasing in \(\tau\).

At the location level, let
\[
N_\ell(\tau)=\sum_{m=1}^M \mathbf{1}\{E_{m\ell}(\tau)>0\}
\]
denote the number of active local specialist bundle-industries. Proposition~\ref{prop:entry} implies that \(N_\ell(\tau)\) is weakly increasing whenever at least one additional bundle is both outsourced and locally delivered.

The model therefore predicts organizational reallocation in two complementary senses. Within incumbent firms, auxiliary production becomes narrower as more peripheral bundles are shifted outside firm boundaries. Across the economy, production becomes broader as more specialist establishments become viable. These comparative statics provide the rationale for interpreting lower secondary production and greater activity in narrow specialist industries as evidence consistent with boundary redrawing.

\subsection{Interpretation}

The mechanism does not require specialist entrants to possess a frontier technological advantage in the underlying task. Their advantage may be organizational. They absorb recruiting, screening, scheduling, payroll, compliance, local coordination, and exception handling into a vendor relationship, thereby reducing the buyer's governance burden. For that reason, the model mainly predicts expansion on the extensive margin, with more specialist establishments serving outsourced bundles. It does not, by itself, imply that average startup quality must rise. In the paper, we therefore treat the absence of startup-quality effects as a qualitative implication of the mechanism rather than as a separately identified theorem.

\subsection{Relation to the Task-Trade Framework}

The structuredness parameter \(s\) in this appendix plays a role analogous to the routineness index \(\mu(t)\) in \citet{costinot_adaptation_2011}. In both frameworks, higher routineness makes arm's-length governance relatively more attractive, because more routine activities are easier to specify and verify ex ante and raise fewer ex post adaptation problems. The central difference is the margin we study. The task-trade literature focuses on the organization of cross-border production and explains why routine activities are more likely to be sourced at arm's length abroad. Our framework uses the same governance logic but adds a local supplier-entry margin. When outsourced bundles require local delivery, the movement of activity outside incumbent firm boundaries can generate domestic specialist entry rather than only foreign sourcing.

The object \(d_{im}=\psi(n_m,s_{im},\tau)\), the dimensionality of the buyer-facing verifiable performance vector, captures the interface-compression channel. It allows the buyer-facing contract to be a compressed representation of the underlying contingency set. The bundle structure, indexed by \(m\) and defined by the \(n_m\) underlying contingencies, also differs from task-trade setups that treat each task as a separable unit of production. These additions are useful for the empirical application because the observed margin is not the intrafirm share of cross-border trade, but the local footprint of domestic supplier formation.

\end{document}

%% file: figures_duha/BA_over_time.tex
\begin{figure}[!htbp]
\caption{Population-Adjusted Business Applications Over Time, 2005-2019}
\label{fig:BA_over_time}
\begin{center}
\includegraphics[width=0.7\textwidth]{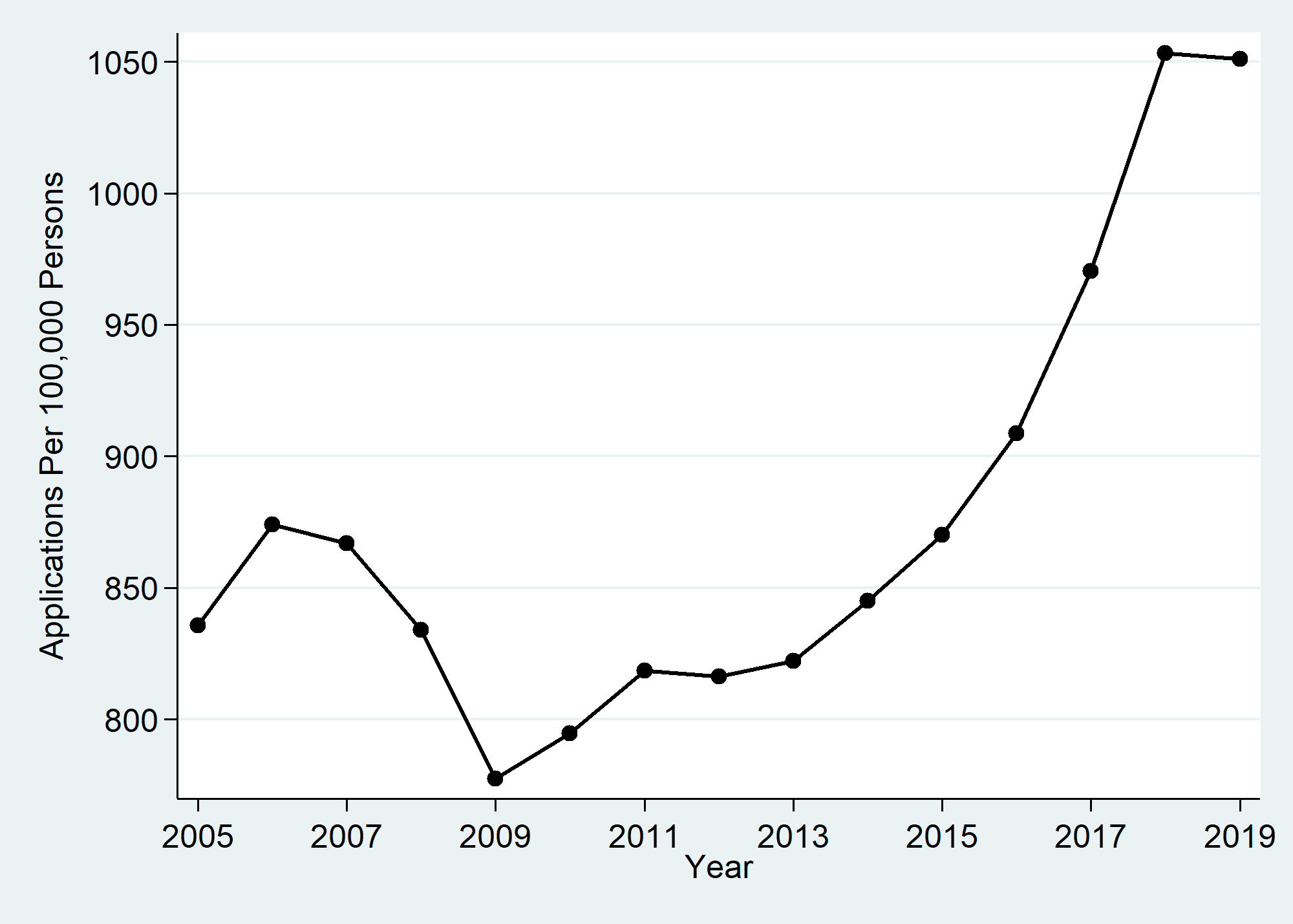}
\end{center}
\scriptsize{\parbox{1\linewidth}{The figure presents the population-weighted annual mean of business applications per 100,000 persons across commuting zones for each year from 2005 to 2019. We weight each commuting zone by its share of the US population in 2000. The sample consists of 722 commuting zones over a 15-year period, resulting in 10,830 commuting zone-year observations.}}
\vspace{10pt}
\end{figure}

%% file: figures_duha/year_by_year_BA_long_dif.tex
\begin{figure}[htbp]
\centering
\caption{Year-by-Year IV Estimates of the Effect of RSH on Business Applications\label{fig_year_by_year_BA_long_dif}}
\includegraphics[width=0.85\textwidth]{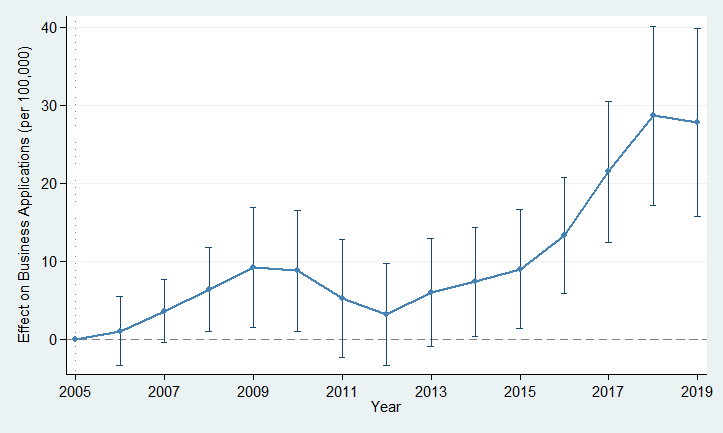}
\begin{minipage}{\textwidth}
\scriptsize
This figure plots 2SLS estimates of the effect of 2005 Routine Employment Share (RSH) on the cumulative change in business applications (per 100{,}000 residents) from 2005 to each subsequent year. Each point corresponds to a separate regression of the change in business applications between year $t$ and 2005 on $RSH_{2005}$, instrumented with the 1990 shift-share instrument. The 2005 coefficient is normalized to zero by construction. Vertical bars show 95\% confidence intervals based on standard errors clustered by state. All specifications include the same baseline controls as in Table~\ref{tab_baseline_results}, state fixed effects, and commuting-zone population weights.
\end{minipage}
\end{figure}

%% file: figures_duha/bds_timing_2sls.tex
\begin{figure}[htbp]
\centering
\caption{Timing Evidence from BDS Establishment Entry\label{fig:bds_timing_annual_2sls}}
\includegraphics[width=0.82\textwidth]{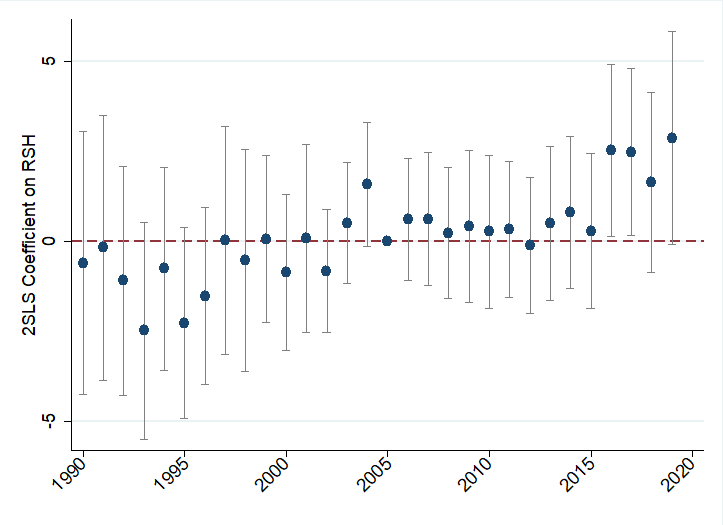}
\par\vspace{6pt}
\begin{minipage}{\textwidth}
\scriptsize
This figure plots annual 2SLS estimates of the relationship between baseline Routine Employment Share ($RSH_{c,2005}$) and cumulative changes in BDS establishment entry per 100{,}000 residents relative to 2005. For years after 2005, the outcome is the change in BDS establishment entry between year $t$ and 2005; for years before 2005, the outcome is the backward cumulative difference between year $t$ and 2005. Each point corresponds to a separate regression using the same baseline controls, state fixed effects, and commuting-zone population weights as in Table~\ref{tab_baseline_results}. $RSH_{c,2005}$ is instrumented with the 1990 shift-share instrument, and vertical bars show 95\% confidence intervals based on standard errors clustered by state. The 2005 coefficient is normalized to zero by construction. Because the pre-2005 points are backward cumulative differences relative to 2005 rather than formal event-study leads, we interpret them as descriptive timing evidence rather than a definitive pretrend test.
\end{minipage}
\end{figure}

%% file: tables_duha/tab_summary_stats.tex
\begin{table}[htbp]
\centering
\caption{Summary Statistics\label{tab:summary_stats}}
\begin{tabular}{p{5.2cm} p{7.0cm} cc}
\hline\hline
Variable & Definition & Mean & Std.\ Dev. \\
\hline
\\[-6pt]
RSH$_{2005}$ &
Routine Employment Share. Employment share (percentage points) in occupations in the top tercile of \cite{deming_growing_2017} routine task intensity in 2005 &
31.4 & 3.1 \\
[6pt]
$\Delta$ Business Applications &
Change in business applications per 100{,}000 residents &
& \\
\hspace{0.8cm}2005--2019 & & 145.8 & 599.1 \\
\hspace{0.8cm}2005--2010 & & $-$52.7 & 167.7 \\
\hspace{0.8cm}2010--2019 & & 198.5 & 577.4 \\
[6pt]
$\Delta$ BDS Establishment Entry &
Change in BDS establishment entries per 100{,}000 residents &
& \\
\hspace{0.8cm}2005--2019 & & $-$39.5 & 59.0 \\
\hspace{0.8cm}2005--2010 & & $-$47.0 & 50.1 \\
\hspace{0.8cm}2010--2019 & & 7.5 & 53.5 \\
[6pt]
$\Delta$ SFR &
Change in startup formation rate between 2005 and 2016 &
264.3 & 1129.2 \\
[6pt]
$\Delta$ EQI &
Change in entrepreneurship quality index between 2005 and 2016 ($\times 100$) &
$-$0.01 & 0.02 \\
[6pt]
$\Delta$ RECPI &
Change in realized entrepreneurial capital per input between 2005 and 2016 ($\times 100$) &
14.4 & 132.6 \\
[6pt]
\hline\hline
\end{tabular}
\par\vspace{6pt}
\begin{minipage}{\textwidth}
\scriptsize
The unit of observation is a commuting zone. RSH, BFS, and BDS summary statistics are based on 722 commuting zones. Startup Cartography outcomes are available for a smaller matched sample of 680 commuting zones.
\end{minipage}
\end{table}

%% file: tables_duha/tab_baseline_results.tex
\begin{table}[htbp]
\centering
\caption{Baseline Results\label{tab_baseline_results}}
\begin{tabular}{lcccc}
\hline\hline
 & (1) & (2) & (3) & (4) \\
 & RSH$_{2005}$ & $\Delta$ Bus.\ App.s & $\Delta$ Bus.\ App.s & $\Delta$ Bus.\ App.s \\
 & (First Stage) & 2005--19 & 2005--10 & 2010--19 \\
\hline
IV & 0.943$^{***}$ & & & \\
 & (0.104) & & & \\
[6pt]
RSH$_{2005}$ & & 27.8$^{***}$ & 8.9$^{**}$ & 19.0$^{***}$ \\
 & & (6.2) & (4.0) & (6.9) \\
\hline
Controls & Yes & Yes & Yes & Yes \\
State FE & Yes & Yes & Yes & Yes \\
Observations & 722 & 722 & 722 & 722 \\
\hline\hline
\end{tabular}
\par\vspace{6pt}
\begin{minipage}{\textwidth}
\scriptsize
Column (1) reports the first-stage regression; columns (2)--(4) report 2SLS estimates. The dependent variable is the change in business applications per 100,000 residents over the indicated period. The excluded instrument is the 1990 shift-share instrument. All specifications include state fixed effects and commuting-zone demographic controls. Robust standard errors clustered by state are in parentheses. Observations are weighted by commuting-zone population share. The Kleibergen--Paap $F$-statistic is 81.8. $^{***}$ $p<0.01$, $^{**}$ $p<0.05$, $^{*}$ $p<0.1$.
\end{minipage}
\end{table}

%% file: tables_duha/BDS_entry.tex
\begin{table}[htbp]
\centering
\caption{Routine Employment Share and Realized Establishment Entry (BDS)\label{tab_BDS}}
\begin{tabular}{l*{3}{c}}
\hline\hline
 & (1) & (2) & (3) \\
 & $\Delta$ Entries & $\Delta$ Entries & $\Delta$ Entries \\
 & 2005--19 & 2005--10 & 2010--19 \\
\hline
\\[-6pt]
RSH$_{2005}$ & 2.87$^{*}$ & 0.27 & 2.60$^{**}$ \\
 & (1.51) & (1.08) & (1.21) \\
[6pt]
\hline
Controls & Yes & Yes & Yes \\
State FE & Yes & Yes & Yes \\
Observations & 722 & 722 & 722 \\
\hline\hline
\end{tabular}
\par\vspace{6pt}
\begin{minipage}{\textwidth}
\scriptsize
2SLS estimates. The dependent variable is the change in total establishment entries per 100{,}000 residents from the Business Dynamics Statistics (BDS) over the indicated period. Kleibergen--Paap $F$-statistic: 81.8. All specifications include state fixed effects and commuting-zone demographic controls. Robust standard errors clustered by state in parentheses. Observations weighted by commuting-zone population share. $^{***}$ $p<0.01$, $^{**}$ $p<0.05$, $^{*}$ $p<0.1$.
\end{minipage}
\end{table}

%% file: tables_duha/table_cbp_size.tex
\begin{table}[htbp]
\centering
\caption{Routine Employment Share and Changes in Establishment Size Shares\label{tab_size_shares}}
\begin{tabular}{l*{4}{c}}
\hline\hline
 & (1) & (2) & (3) & (4) \\
 & 1--4 & 5--19 & 20--49 & 50+ \\
\hline
\\[-6pt]
\multicolumn{5}{c}{\textit{Panel A: 2005--2019}} \\
\hline
RSH$_{2005}$ & 0.197$^{**}$ & $-$0.111$^{**}$ & $-$0.037 & $-$0.050$^{**}$ \\
 & (0.077) & (0.046) & (0.030) & (0.022) \\
\hline
\\[-6pt]
\multicolumn{5}{c}{\textit{Panel B: 2005--2010}} \\
\hline
RSH$_{2005}$ & 0.039 & 0.014 & 0.000 & $-$0.053$^{***}$ \\
 & (0.047) & (0.031) & (0.016) & (0.015) \\
\hline
\\[-6pt]
\multicolumn{5}{c}{\textit{Panel C: 2010--2019}} \\
\hline
RSH$_{2005}$ & 0.159$^{***}$ & $-$0.125$^{***}$ & $-$0.037 & 0.003 \\
 & (0.052) & (0.041) & (0.023) & (0.016) \\
\hline
Controls & Yes & Yes & Yes & Yes \\
State FE & Yes & Yes & Yes & Yes \\
Observations & 722 & 722 & 722 & 722 \\
\hline\hline
\end{tabular}
\par\vspace{6pt}
\begin{minipage}{\textwidth}
\scriptsize
2SLS estimates using the same instrument and baseline controls as in Table~\ref{tab_baseline_results}. The dependent variables are long differences in the share of establishments in each size category, multiplied by 100, so coefficients are percentage-point changes in establishment shares associated with a 1 percentage-point increase in $RSH_{2005}$. The size bins are 1--4, 5--19, 20--49, and 50+ employees. The 50+ category is constructed residually from County Business Patterns (CBP) total establishment counts. All specifications include state fixed effects and baseline commuting-zone controls. Standard errors clustered by state in parentheses. Observations weighted by commuting-zone population share. Because the four size shares sum to one, coefficients within a panel sum to approximately zero, up to rounding. Kleibergen--Paap $F$-statistic: 81.8. $^{***}$ $p<0.01$, $^{**}$ $p<0.05$, $^{*}$ $p<0.1$.
\end{minipage}
\end{table}

%% file: tables_duha/table_cartography.tex
\begin{table}[htbp]
\centering
\caption{Routine Employment Share and Entrepreneurial Quality in the Cartography Sample\label{tab:cartography}}
\begin{tabular}{l*{4}{c}}
\hline\hline
 & (1) & (2) & (3) & (4) \\
 & $\Delta$ BFS Apps & $\Delta$ SFR & $\Delta$ EQI & $\Delta$ RECPI \\
\hline
\\[-6pt]
RSH$_{2005}$ & 14.1$^{***}$ & 293.6 & $-$0.0000182$^{**}$ & 0.339 \\
 & (3.93) & (224.3) & (0.0000092) & (0.253) \\
[6pt]
\hline
Controls & Yes & Yes & Yes & Yes \\
State FE & Yes & Yes & Yes & Yes \\
Observations & 680 & 680 & 680 & 680 \\
\hline\hline
\end{tabular}
\par\vspace{6pt}
\begin{minipage}{\textwidth}
\scriptsize
2SLS estimates. Column~(1) re-estimates the BFS business-application specification on the matched sample of 680 commuting zones covered by the Startup Cartography data, with the dependent variable equal to the change in business applications per 100{,}000 residents between 2005 and 2016. Columns~(2)--(4) report estimates for changes between 2005 and 2016 in the Startup Formation Rate (SFR), the Entrepreneurial Quality Index (EQI), and realized entrepreneurial capital per input (RECPI). All specifications include state fixed effects and baseline commuting-zone controls. Standard errors clustered by state in parentheses. Observations weighted by commuting-zone population share. The first stage is common across columns. Kleibergen--Paap $F$-statistic: 93.9. $^{***}$ $p<0.01$, $^{**}$ $p<0.05$, $^{*}$ $p<0.1$.
\end{minipage}
\end{table}

%% file: tables_duha/horse_race.tex
\begin{table}[htbp]
\centering
\caption{Routine Cognitive versus Routine Manual Exposure\label{tab:horse_race_tasks}}
\begin{tabular}{lccc}
\hline\hline
 & (1) & (2) & (3) \\
 & $\Delta$ Bus.\ App.s & $\Delta$ Bus.\ App.s & $\Delta$ Bus.\ App.s \\
 & 2005--19 & 2005--10 & 2010--19 \\
\hline
\\[-6pt]
Routine Cognitive & 43.1$^{**}$ & 26.1$^{***}$ & 16.9 \\
 & (16.8) & (6.2) & (16.4) \\
[6pt]
Routine Manual & $-$5.9 & 6.0 & $-$12.0$^{*}$ \\
 & (8.7) & (5.0) & (6.5) \\
\hline
Controls & Yes & Yes & Yes \\
State FE & Yes & Yes & Yes \\
Observations & 722 & 722 & 722 \\
\hline\hline
\end{tabular}
\par\vspace{6pt}
\begin{minipage}{\textwidth}
\scriptsize
2SLS estimates with two endogenous regressors: routine cognitive task intensity and routine manual task intensity. Each is instrumented using a separate leave-one-state-out shift-share instrument constructed analogously to the baseline IV from 1990 industry shares and industry-level routine cognitive/manual employment shares. All specifications include state fixed effects and baseline commuting-zone controls, use commuting-zone population weights, and report standard errors clustered by state. Sanderson--Windmeijer first-stage $F$-statistics for the two endogenous regressors are 28.7 and 27.7; the Kleibergen--Paap weak-identification statistic is 12.6. $^{***}$ $p<0.01$, $^{**}$ $p<0.05$, $^{*}$ $p<0.1$.
\end{minipage}
\end{table}

%% file: tables_duha/tab_control_robots_china.tex
\begin{table}[htbp]
\centering
\caption{Robustness to Controlling for Robot Exposure\label{tab_control_robots_china}}
\begin{tabular}{l*{3}{c}}
\hline\hline
 & (1) & (2) & (3) \\
 & $\Delta$ Bus.\ App.s & $\Delta$ Bus.\ App.s & $\Delta$ Bus.\ App.s \\
 & 2005--19 & 2005--10 & 2010--19 \\
\hline
RSH$_{2005}$ & 25.3$^{***}$ & 7.8$^{**}$ & 17.4$^{**}$ \\
 & (6.1) & (3.7) & (7.1) \\
[6pt]
Robot Exposure & 13.9$^{***}$ & 5.4$^{**}$ & 8.5$^{**}$ \\
 & (4.1) & (2.7) & (3.4) \\
\hline
Controls & Yes & Yes & Yes \\
State FE & Yes & Yes & Yes \\
Observations & 722 & 722 & 722 \\
\hline\hline
\end{tabular}
\par\vspace{6pt}
\begin{minipage}{\textwidth}
\scriptsize
2SLS estimates. Each column re-estimates the baseline specification while controlling for robot exposure. RSH$_{2005}$ is instrumented using the 1990 shift-share instrument, and robot exposure is instrumented using European robot adoption following \citet{acemoglu_robots_2020}. All specifications include state fixed effects and baseline commuting-zone controls. Standard errors clustered by state in parentheses. Observations weighted by commuting-zone population share. The Sanderson--Windmeijer first-stage $F$-statistic for RSH$_{2005}$ is 95.4; for robot exposure, 1156.3. $^{***}$ $p<0.01$, $^{**}$ $p<0.05$, $^{*}$ $p<0.1$.
\end{minipage}
\end{table}

%% file: tables_duha/table_combined_io_organizational.tex
\begin{table}[htbp]
\centering
\caption{Codifiability Decomposition of Organizational Breadth and Production Structure\label{tab:combined_io_org}}
\begin{tabular}{lcccc}
\hline\hline
\\[-6pt]
\multicolumn{5}{c}{\textit{Panel A: $\Delta$ Active 4-digit NAICS codes}} \\
\hline
 & (1) & (2) & (3) & (4) \\
\hline
RSH$_{2005}$ & 5.75$^{***}$ & & & \\
 & (0.86) & & & \\
[6pt]
Routine Cognitive & & 9.51$^{***}$ & & 7.46$^{***}$ \\
 & & (1.55) & & (1.81) \\
[6pt]
Routine Manual & & & $-$4.17$^{***}$ & $-$1.59 \\
 & & & (0.58) & (1.12) \\
\hline
\\[-6pt]
\multicolumn{5}{c}{\textit{Panel B: $\Delta$ Secondary production share}} \\
\hline
 & (1) & (2) & (3) & (4) \\
\hline
RSH$_{2005}$ & $-$0.0020$^{***}$ & & & \\
 & (0.0004) & & & \\
[6pt]
Routine Cognitive & & $-$0.0031$^{***}$ & & $-$0.0038$^{***}$ \\
 & & (0.0006) & & (0.0007) \\
[6pt]
Routine Manual & & & 0.0008$^{***}$ & $-$0.0006 \\
 & & & (0.0003) & (0.0004) \\
\hline
Controls & Yes & Yes & Yes & Yes \\
State FE & Yes & Yes & Yes & Yes \\
Observations & 720 & 720 & 720 & 720 \\
\hline\hline
\end{tabular}
\par\vspace{6pt}
\begin{minipage}{\textwidth}
\scriptsize
2SLS long-difference estimates (2005--2019). Panel~A counts 4-digit NAICS industries with at least one establishment in County Business Patterns; a positive coefficient indicates that higher-exposure commuting zones gained more active industry categories, consistent with organizational broadening. Panel~B measures the employment-weighted share of CZ employment in industries producing outside their primary commodity classification (Make Table off-diagonal share, BEA 2007 I-O Tables, Before Redefinitions, Detail Level); a negative coefficient indicates a shift toward more specialized production, consistent with firms shedding peripheral activities. Column~(1) instruments RSH with the baseline shift-share IV. Columns~(2) and (3) instrument routine cognitive and routine manual employment shares separately with split leave-one-state-out Bartik IVs. Column~(4) treats both as endogenous. All specifications include state fixed effects, baseline controls, and population-share weights. Standard errors clustered by state. Kleibergen--Paap $F$-statistics: col.~(1), 81.8; col.~(2), 110.0; col.~(3), 56.7; col.~(4), 12.6. $^{***}p<0.01$, $^{**}p<0.05$, $^{*}p<0.10$.
\end{minipage}
\end{table}

%% file: tables_duha/contract_projection.tex
\begin{landscape}
\begin{table}[htbp]
\centering
\caption{Projected Exposure to Contract Compression and Routine Employment Share\label{tab:contract_projection}}
\begin{tabular}{lcccccccc}
\hline\hline
 & \multicolumn{2}{c}{Performance/Deliverables} & \multicolumn{2}{c}{Automation/Technology} & \multicolumn{2}{c}{Composite Index} & \multicolumn{2}{c}{Placebo (Credit)} \\
\cmidrule(lr){2-3} \cmidrule(lr){4-5} \cmidrule(lr){6-7} \cmidrule(lr){8-9}
 & (1) & (2) & (3) & (4) & (5) & (6) & (7) & (8) \\
\hline
\\[-6pt]
RSH$_{2005}$ & 0.00122$^{***}$ & & 0.00174$^{**}$ & & 0.01285$^{**}$ & & 0.00075 & \\
 & (0.00029) & & (0.00068) & & (0.00615) & & (0.00067) & \\
[6pt]
Routine Cognitive & & 0.00163$^{**}$ & & 0.00239$^{*}$ & & 0.03216$^{***}$ & & 0.00179 \\
 & & (0.00079) & & (0.00132) & & (0.01002) & & (0.00173) \\
[6pt]
Routine Manual & & 0.00010 & & $-$0.00003 & & 0.00395 & & $-$0.00031 \\
 & & (0.00038) & & (0.00064) & & (0.00776) & & (0.00088) \\
\hline
Controls & Yes & Yes & Yes & Yes & Yes & Yes & Yes & Yes \\
State FE & Yes & Yes & Yes & Yes & Yes & Yes & Yes & Yes \\
Observations & 722 & 722 & 722 & 722 & 722 & 722 & 722 & 722 \\
\hline\hline
\end{tabular}
\par\vspace{6pt}
\begin{minipage}{\linewidth}
\scriptsize
2SLS estimates of commuting-zone-level projections of industry-level changes in contract language on baseline routine task exposure. The first column of each pair instruments Routine Employment Share with the 1990 leave-one-state-out Bartik IV used throughout the paper. The second column of each pair jointly instruments routine cognitive and routine manual exposure with their respective split Bartik IVs. The contract-language shifts are measured from SEC EDGAR Exhibit~10 filings between 2004--2006 and 2017--2019 and projected to commuting zones using 2005 industry employment shares (approximately three-digit NAICS). Performance/Deliverables uses performance-and-deliverables keywords. Automation/Technology uses automation-and-technology keywords. Composite Index is a z-score composite of these two outcomes, the credit-minus-B2B word-count gap, the tier-2 vendor-management keyword share, and B2B contract volume per industry employee. The placebo is built from credit-agreement word counts. All specifications include state fixed effects, baseline controls, and commuting-zone population weights. Standard errors clustered by state in parentheses. Kleibergen--Paap $F$-statistic is 81.8 for the RSH columns and 12.6 for the horse-race columns. $^{***}p<0.01$, $^{**}p<0.05$, $^{*}p<0.10$.
\end{minipage}
\end{table}
\end{landscape}

%% file: figures_duha/DEMING_map.tex
\begin{figure}[htbp]
\centering
\caption{Geographic Variation in Routine Employment Share, 2005\label{fig:aes_map}}
\includegraphics[width=\textwidth]{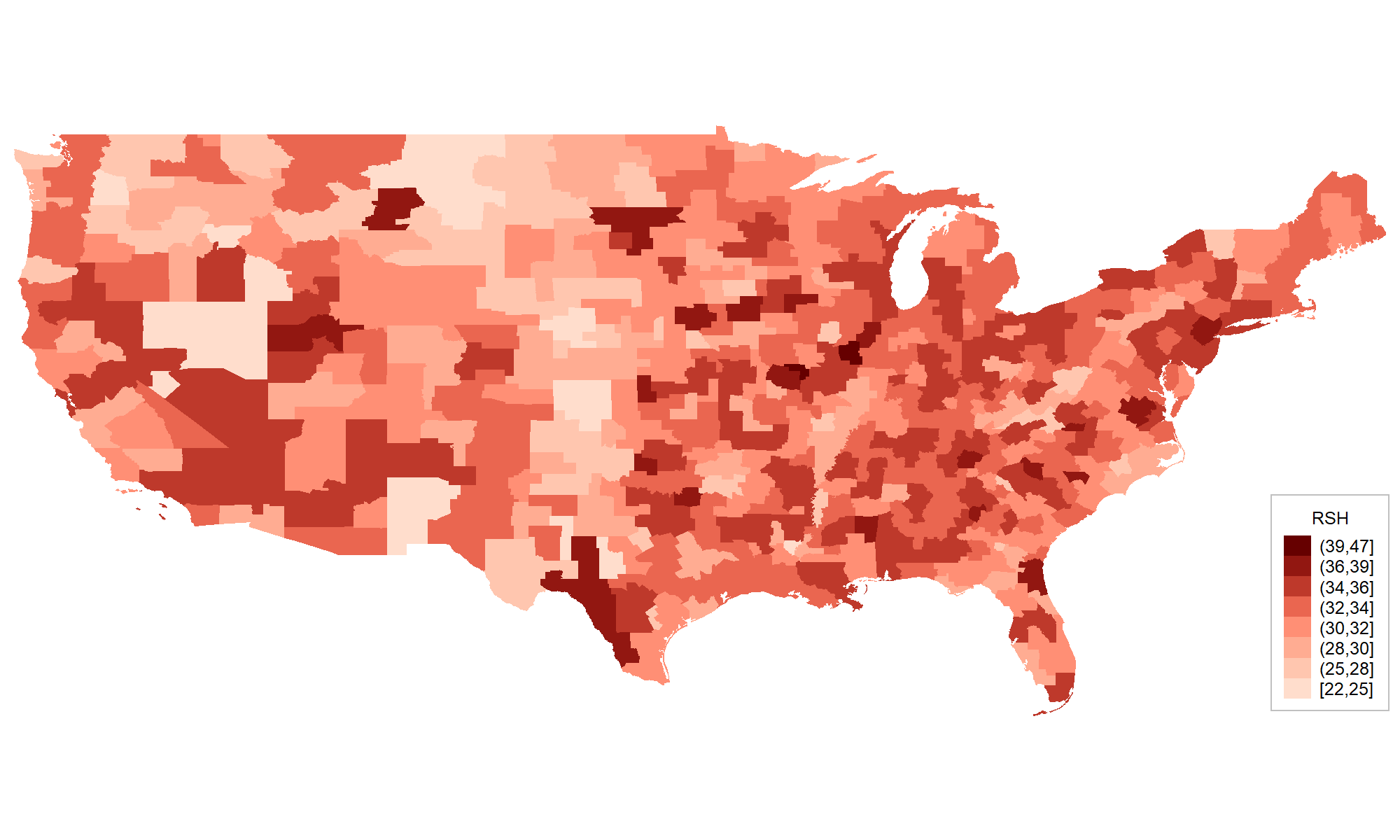}
\begin{minipage}{\textwidth}
\scriptsize
This figure maps Routine Employment Share (RSH) across U.S. commuting zones in 2005. RSH is defined as the employment share (percentage points) in occupations in the top tercile of the \citet{deming_growing_2017} routine task intensity distribution. Darker shades indicate greater exposure.
\end{minipage}
\end{figure}

%% file: figures_duha/ba_app_map.tex
\begin{figure}[htbp]
\centering
\caption{Geographic Variation in Business Application Growth, 2005--2019}
\label{fig:busapps_change_map}
\includegraphics[width=\textwidth]{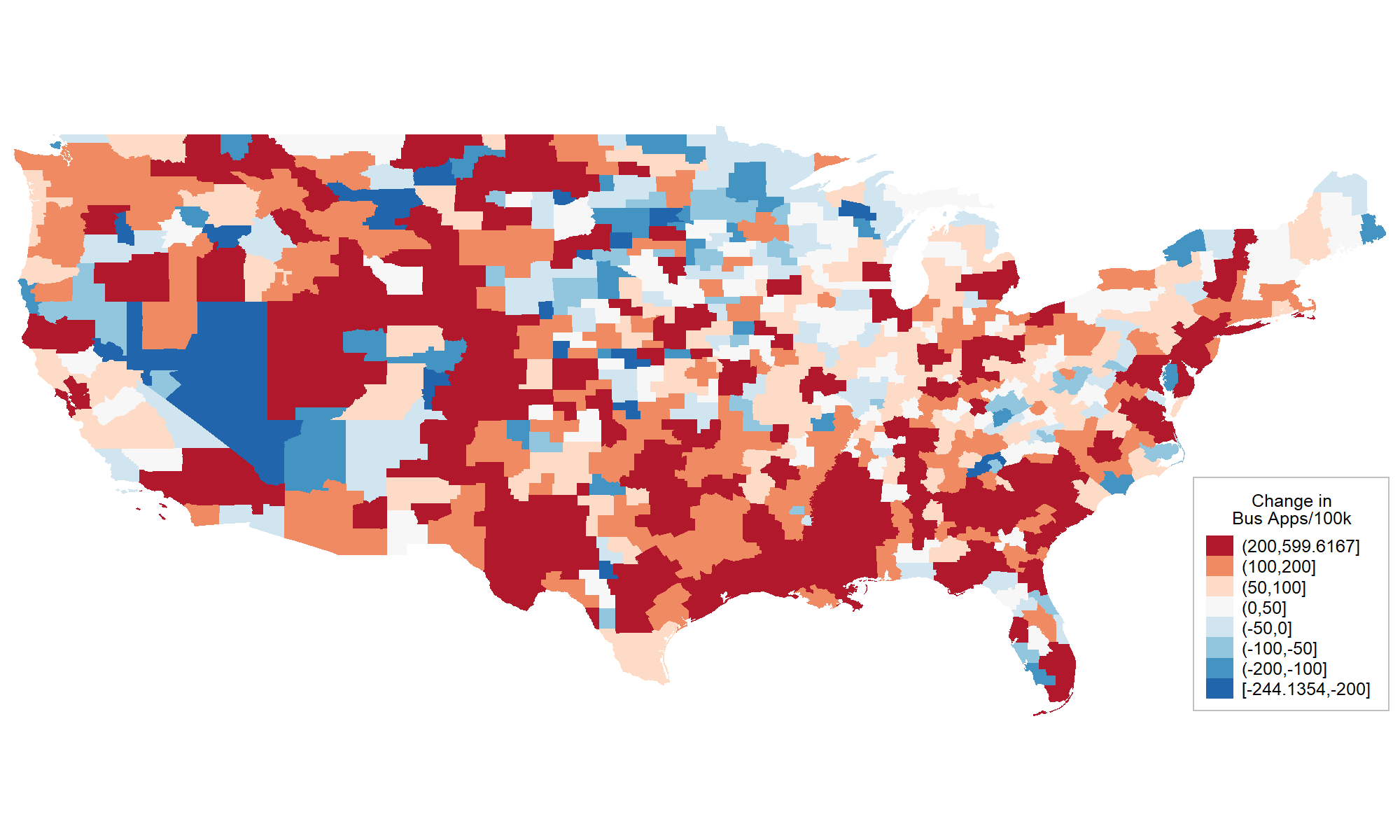}
\begin{minipage}{0.95\textwidth}
\small
\textit{Notes:} This figure maps the change in business applications per 100{,}000 residents across U.S. commuting zones between 2005 and 2019. Red shades indicate larger increases in business applications, while blue shades indicate declines.
\end{minipage}
\end{figure}

%% file: figures_duha/leave_one_out.tex
\begin{figure}[htbp]
\centering
\caption{Leave-One-State-Out IV Estimates\label{fig:leave_one_out}}
\includegraphics[width=0.85\textwidth]{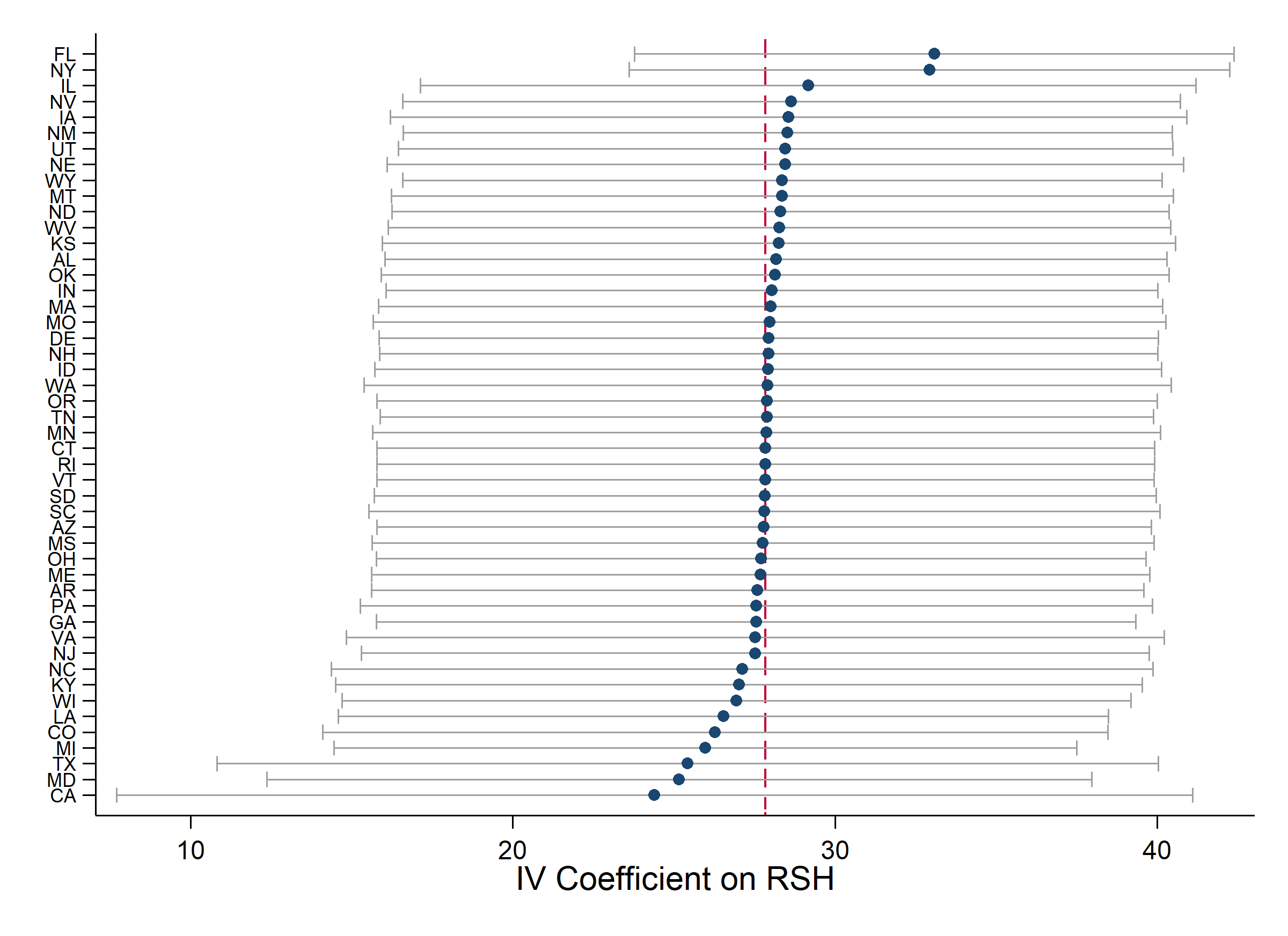}
\par\vspace{6pt}
\begin{minipage}{\textwidth}
\scriptsize
Each point reports the 2SLS coefficient on RSH from the baseline specification after excluding all commuting zones in the indicated state. Horizontal lines show 95\% confidence intervals based on state-clustered standard errors. The dashed vertical line marks the full-sample baseline estimate. All specifications include state fixed effects, baseline CZ controls, and CZ population weights.
\end{minipage}
\end{figure}

%% file: figures_duha/bds_timing_rf.tex
\begin{figure}[htbp]
\centering
\caption{Reduced-Form Timing Evidence from BDS Establishment Entry\label{fig:bds_timing_annual_rf}}
\includegraphics[width=0.82\textwidth]{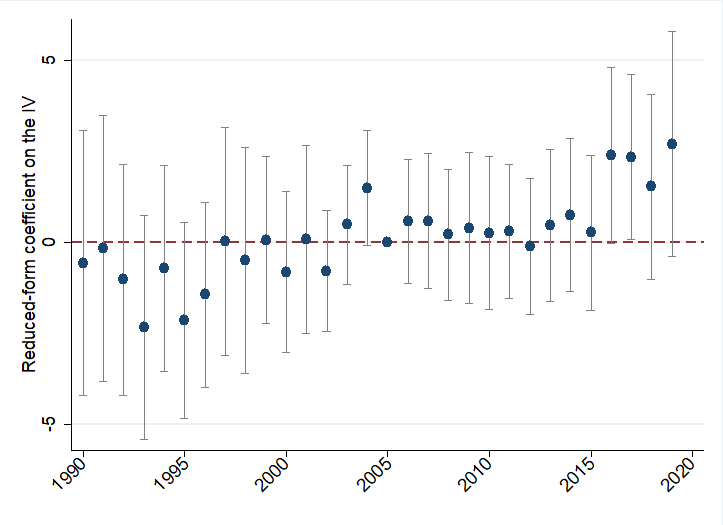}
\par\vspace{6pt}
\begin{minipage}{\textwidth}
\scriptsize
This figure plots annual reduced-form estimates of the relationship between the 1990 shift-share instrument and cumulative changes in BDS establishment entry per 100{,}000 residents relative to 2005. For years after 2005, the outcome is the change in BDS establishment entry between year $t$ and 2005; for years before 2005, the outcome is the backward cumulative difference between year $t$ and 2005. Each point corresponds to a separate regression using the same baseline controls, state fixed effects, and commuting-zone population weights as in Table~\ref{tab_baseline_results}. Vertical bars show 95\% confidence intervals based on standard errors clustered by state. The 2005 coefficient is normalized to zero by construction. As in Figure~\ref{fig:bds_timing_annual_2sls}, the pre-2005 points are descriptive timing evidence rather than formal event-study leads.
\end{minipage}
\end{figure}

%% file: figures_duha/bds_timing_5year_2sls.tex
\begin{figure}[htbp]
\centering
\caption{Five-Year Timing Estimates for BDS Establishment Entry\label{fig:bds_timing_5yr_2sls}}
\includegraphics[width=0.78\textwidth]{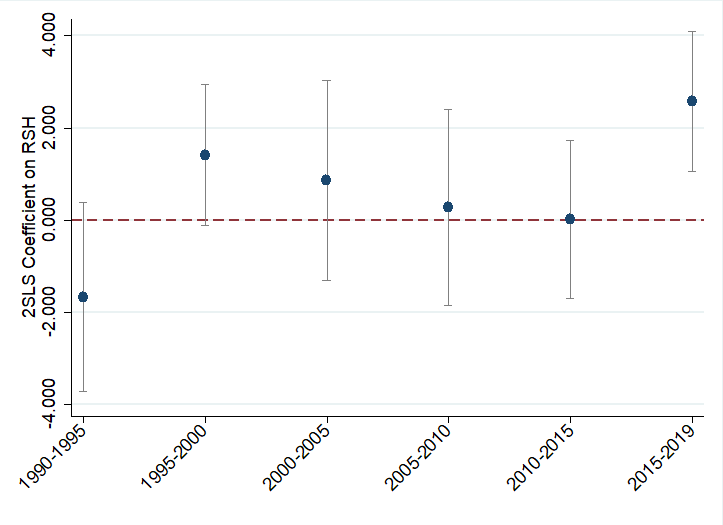}
\par\vspace{6pt}
\begin{minipage}{\textwidth}
\scriptsize
This figure plots 2SLS estimates of the relationship between baseline Routine Employment Share ($RSH_{c,2005}$) and non-overlapping five-year changes in BDS establishment entry per 100{,}000 residents. The estimation windows are 1990--1995, 1995--2000, 2000--2005, 2005--2010, 2010--2015, and 2015--2019. Each point corresponds to a separate regression using the same baseline controls, state fixed effects, and commuting-zone population weights as in Table~\ref{tab_baseline_results}. $RSH_{c,2005}$ is instrumented with the 1990 shift-share instrument, and vertical bars show 95\% confidence intervals based on standard errors clustered by state. Because these exercises rely on a different outcome measure than the baseline BFS regressions and apply common baseline controls across windows, we interpret them as supportive timing and placebo evidence rather than definitive pretrend tests.
\end{minipage}
\end{figure}

%% file: tables_duha/tab_aes_correlations.tex
\begin{table}[htbp]
\centering
\caption{Correlation of RSH with Other Metrics\label{tab:aes_correlations}}
\begin{tabular}{p{9.0cm} c}
\hline\hline
Task Measure & Correlation with RSH \\
\hline
\\[-6pt]
Degree of Job Automation$^{c}$ & 0.593 \\
Routine Cognitive$^{b}$ & 0.708 \\
Importance of Repeating$^{c}$ & 0.854 \\
Cognitive Tasks$^{a}$ & 0.316 \\
Social Tasks$^{a}$ & 0.168 \\
Non-Routine Cognitive (Analytical)$^{b}$ & 0.171 \\
Non-Routine Cognitive (Interpersonal)$^{b}$ & 0.032 \\
Routine Manual$^{b}$ & $-$0.294 \\[6pt]
\hline
Observations & 722 \\
\hline\hline
\end{tabular}
\par\vspace{6pt}
\begin{minipage}{\textwidth}
\scriptsize
Pearson correlations between RSH (Routine Employment Share) and alternative task-based employment share measures at the commuting zone level in 2005. $^{a}$~\cite{deming_growing_2017}. \quad $^{b}$~\cite{acemoglu_skills_2011}. \quad $^{c}$~O*NET Work Context.
\end{minipage}
\end{table}

%% file: tables_duha/tab_robustness.tex
{\setlength{\tabcolsep}{4pt}
\begin{center}
\begin{longtable}{@{}l*{4}{>{\centering\arraybackslash}p{0.15\linewidth}}@{}}
\caption{Robustness Tests\label{tab:robustness}} \\
\hline\hline
\endfirsthead
\multicolumn{5}{l}{\textit{Table \ref{tab:robustness} (continued)}} \\
\hline
\endhead
\hline
\endfoot
\hline
\multicolumn{5}{p{0.97\linewidth}}{\scriptsize Col.\ (1) reports first-stage coefficients in the IV panels; cols.\ (2)--(4) report 2SLS estimates, except Panel~B, which reports OLS. The dependent variable is the change in business applications per 100{,}000 residents (log change in Panel~D). Panel~A uses the 1980 instrument; Panel~C uses quartile controls; Panel~E additionally controls for a leave-one-out Bartik measure of Great Recession exposure, constructed from 2005 industry employment shares and 2007--2010 national industry employment growth; Panel~F varies the RSH cutoff. Panel~G reports the baseline point estimates with alternative inference procedures: heteroskedastic-robust standard errors, census-division clustering with Webb wild-bootstrap $p$-values, and spatial standard errors using a 100-mile cutoff. Unless otherwise noted, all specifications include state fixed effects and baseline CZ controls, use CZ population weights, and report state-clustered standard errors. Panel~B uses 720 observations. Other panels use 722 observations, except the spatial estimates in Panel~G, which use 721 because one CZ lacks centroid coordinates. Kleibergen--Paap $F$-statistics are: Panel~A, 91.7; Panel~C, 101.8; Panel~D, 81.8; Panel~E, 78.7; Panel~F, 85.3 (50th), 73.3 (75th), 163.2 (90th). $^{***}$ $p<0.01$, $^{**}$ $p<0.05$, $^{*}$ $p<0.1$.} \\
\endlastfoot
\\[-6pt]
\multicolumn{5}{c}{\textit{Panel A: 1980 Instrument}} \\*
\hline
 & (1) & (2) & (3) & (4) \\*
 & RSH$_{2005}$ & $\Delta$ Bus.\ App.s & $\Delta$ Bus.\ App.s & $\Delta$ Bus.\ App.s \\*
 & (First Stage) & 2005--19 & 2005--10 & 2010--19 \\*
\hline
IV & 0.709$^{***}$ & & & \\*
 & (0.074) & & & \\*
[6pt]
RSH$_{2005}$ & & 34.2$^{***}$ & 7.2$^{*}$ & 27.0$^{***}$ \\*
 & & (6.7) & (4.0) & (6.9) \\*
\hline
Observations & 722 & 722 & 722 & 722 \\
\hline
\\[-6pt]
\multicolumn{5}{c}{\textit{Panel B: OLS}} \\*
\hline
 & (1) & (2) & (3) &  \\*
 & $\Delta$ Bus.\ App.s & $\Delta$ Bus.\ App.s & $\Delta$ Bus.\ App.s &  \\*
 & 2005--19 & 2005--10 & 2010--19 &  \\*
\hline
RSH$_{2005}$ & 15.0$^{***}$ & 2.8 & 12.2$^{***}$ &  \\*
 & (4.6) & (2.3) & (3.8) &  \\*
\hline
Observations & 720 & 720 & 720 &  \\
\hline
\\[-6pt]
\multicolumn{5}{c}{\textit{Panel C: Quartile Controls}} \\*
\hline
 & (1) & (2) & (3) & (4) \\*
 & RSH$_{2005}$ & $\Delta$ Bus.\ App.s & $\Delta$ Bus.\ App.s & $\Delta$ Bus.\ App.s \\*
 & (First Stage) & 2005--19 & 2005--10 & 2010--19 \\*
\hline
IV & 0.661$^{***}$ & & & \\*
 & (0.065) & & & \\*
[6pt]
RSH$_{2005}$ & & 48.3$^{***}$ & 26.9$^{***}$ & 21.4$^{***}$ \\*
 & & (10.0) & (7.6) & (5.2) \\*
\hline
Observations & 722 & 722 & 722 & 722 \\
\newpage
\\[-6pt]
\multicolumn{5}{c}{\textit{Panel D: Log Outcome}} \\*
\hline
 & (1) & (2) & (3) & (4) \\*
 & RSH$_{2005}$ & $\Delta \log$ Bus.\ App.s & $\Delta \log$ Bus.\ App.s & $\Delta \log$ Bus.\ App.s \\*
 & (First Stage) & 2005--19 & 2005--10 & 2010--19 \\*
\hline
IV & 0.943$^{***}$ & & & \\*
 & (0.104) & & & \\*
[6pt]
RSH$_{2005}$ & & 0.029$^{***}$ & 0.010$^{**}$ & 0.019$^{***}$ \\*
 & & (0.005) & (0.004) & (0.004) \\*
\hline
Observations & 722 & 722 & 722 & 722 \\
\hline
\\[-6pt]
\multicolumn{5}{c}{\textit{Panel E: Recession Exposure Control}} \\*
\hline
 & (1) & (2) & (3) & (4) \\*
 & RSH$_{2005}$ & $\Delta$ Bus.\ App.s & $\Delta$ Bus.\ App.s & $\Delta$ Bus.\ App.s \\*
 & (First Stage) & 2005--19 & 2005--10 & 2010--19 \\*
\hline
IV & 0.946$^{***}$ & & & \\*
 & (0.107) & & & \\*
[6pt]
RSH$_{2005}$ & & 28.4$^{***}$ & 9.2$^{**}$ & 19.2$^{***}$ \\*
 & & (5.9) & (3.9) & (6.7) \\*
\hline
Observations & 722 & 722 & 722 & 722 \\
\newpage
\\[-6pt]
\multicolumn{5}{c}{\textit{Panel F: Alternative RSH Thresholds}} \\*
\hline
 & (1) & (2) & (3) & (4) \\*
 & RSH & $\Delta$ Bus.\ App.s & $\Delta$ Bus.\ App.s & $\Delta$ Bus.\ App.s \\*
 & (First Stage) & 2005--19 & 2005--10 & 2010--19 \\*
\hline
\multicolumn{5}{l}{\textit{50th percentile}} \\*[2pt]
IV & 0.830$^{***}$ & & & \\*
 & (0.090) & & & \\*
[6pt]
RSH$^{p50}$ & & 31.7$^{***}$ & 10.1$^{**}$ & 21.6$^{***}$ \\*
 & & (7.4) & (4.4) & (8.2) \\*
[6pt]
\hline
\multicolumn{5}{l}{\textit{75th percentile}} \\*[2pt]
IV & 0.860$^{***}$ & & & \\*
 & (0.100) & & & \\*
[6pt]
RSH$^{p75}$ & & 30.6$^{***}$ & 9.7$^{**}$ & 20.8$^{***}$ \\*
 & & (7.3) & (4.5) & (7.8) \\*
[6pt]
\hline
\multicolumn{5}{l}{\textit{90th percentile}} \\*[2pt]
IV & 0.458$^{***}$ & & & \\*
 & (0.036) & & & \\*
[6pt]
RSH$^{p90}$ & & 57.3$^{***}$ & 18.2$^{**}$ & 39.1$^{**}$ \\*
 & & (14.7) & (7.6) & (16.2) \\*
\hline
Observations & 722 & 722 & 722 & 722 \\
\newpage
\\[-6pt]
\multicolumn{5}{c}{\textit{Panel G: Alternative Standard Errors}} \\*
\hline
 & (1) & (2) & (3) & (4) \\*
 & RSH$_{2005}$ & $\Delta$ Bus.\ App.s & $\Delta$ Bus.\ App.s & $\Delta$ Bus.\ App.s \\*
 & (First Stage) & 2005--19 & 2005--10 & 2010--19 \\*
\hline
Point estimate & 0.943 & 27.8 & 8.9 & 19.0 \\*
[3pt]
Robust SE & (0.081) & (6.239) & (3.477) & (5.311) \\*
$p$-value & [0.000] & [0.000] & [0.011] & [0.000] \\*
[6pt]
Census-division cluster SE & (0.124) & (4.931) & (4.809) & (7.640) \\*
Wild-bootstrap $p$-value & --- & [0.001] & [0.122] & [0.057] \\*
[6pt]
Spatial SE (100-mile cutoff) & --- & (6.717) & (3.776) & (5.614) \\*
$p$-value & --- & [0.000] & [0.019] & [0.001] \\
\hline
\end{longtable}
\end{center}
}

%% file: tables_duha/policy_balance.tex
\begin{table}[htbp]
\centering
\caption{Balance Tests: Instruments and Policy Changes (2005--2019)\label{table_policy_balance}}
\begin{tabular}{lcccc}
\hline\hline
 & (1) & (2) & (3) & (4) \\
 & $\Delta$ RTW & $\Delta$ Corp.\ Tax & $\Delta$ Min.\ Wage & $\Delta$ NCA \\
\hline
\\[-6pt]
IV & 0.026 & 0.032 & 0.012 & 0.001 \\
 & (0.016) & (0.221) & (0.115) & (0.002) \\
[6pt]
\hline
Observations & 722 & 722 & 722 & 722 \\
\hline\hline
\end{tabular}
\par\vspace{6pt}
\begin{minipage}{\textwidth}
\scriptsize
Each column regresses a 2005--2019 state policy change on the 1990 routine share instrument. RTW is right-to-work status; Corp.\ Tax is the state corporate tax rate; Min.\ Wage is the effective minimum wage; NCA is non-compete agreement enforceability. Null coefficients indicate the instrument does not predict state policy changes, supporting the exclusion restriction. All specifications include baseline commuting-zone controls and are weighted by commuting-zone population share. Standard errors clustered by state in parentheses. $^{***}$ $p<0.01$, $^{**}$ $p<0.05$, $^{*}$ $p<0.1$.
\end{minipage}
\end{table}

%% file: tables_duha/tab_rotemberg.tex
\begin{table}[htbp]
\centering
\caption{Top Rotemberg-Weight Industries\label{tab:rotemberg_top10}}
\small
\begin{tabular}{c l c c c c c}
\hline\hline
 & & & & Share of & Just-ID & First stage \\
Rank & Industry & $g_k$ & $\hat{\alpha}_k$ & $|\hat{\alpha}_k|$ (\%) & $\hat{\beta}_k$ & $\hat{\pi}_k$ \\
\hline
1 & Insurance & 0.404 & 0.2774 & 16.0 & 32.57 & 105.10 \\
2 & Agr.: crops & $-$0.277 & 0.1559 & 9.0 & 10.27 & $-$52.49 \\
3 & Mining: coal & $-$0.258 & 0.1049 & 6.0 & 17.33 & $-$31.81 \\
4 & Construction & $-$0.250 & 0.0967 & 5.6 & 70.52 & $-$41.95 \\
5 & Agr.: livestock & $-$0.279 & 0.0738 & 4.3 & 5.02 & $-$120.76 \\
6 & Banking & 0.397 & 0.0651 & 3.8 & 23.32 & 83.48 \\
7 & Mfg: knitting & 0.235 & 0.0432 & 2.5 & 41.37 & 83.37 \\
8 & Legal services & 0.585 & 0.0427 & 2.5 & 42.87 & 157.82 \\
9 & Security brokers & 0.336 & 0.0412 & 2.4 & 0.54 & 128.24 \\
10 & Hospitals & $-$0.128 & $-$0.0407 & 2.3 & 31.03 & 53.22 \\
\hline
\end{tabular}
\par\vspace{6pt}
\begin{minipage}{\textwidth}
\scriptsize
This table reports the 10 industries with the largest absolute Rotemberg weights for the common-shock approximation to the baseline shift-share instrument. ``Share of $|\hat{\alpha}_k|$'' normalizes absolute weights by total absolute Rotemberg mass, $\sum_k |\hat{\alpha}_k| = 1.7351$. The top 1, top 5, and top 10 industries account for 16.0, 40.8, and 54.3 percent of total absolute Rotemberg mass, respectively. In the CZ-level diagnostics file, the common-shock approximation is correlated 0.9994 with the exact leave-one-state-out instrument.
\end{minipage}
\end{table}

%% file: tables_duha/bartik_sensitivity.tex
\begin{table}[htbp]
\centering
\caption{Sensitivity of the Baseline 2SLS Estimate to Top-$k$ and Drop-$k$ Sector Instruments\label{tab:bartik_exact_sensitivity}}
\begin{tabular}{lccc}
\hline\hline
 & Coefficient & & \\
Variant & on RSH & s.e. & KP $F$ \\
\hline
Baseline & 27.8$^{***}$ & (6.2) & 81.8 \\
\hline
\multicolumn{4}{c}{\textit{Panel A: Instruments built using only the top-$k$ industries}} \\
\hline
$k=1$ & 32.6$^{***}$ & (7.6) & 55.6 \\
$k=2$ & 34.4$^{***}$ & (8.2) & 57.6 \\
$k=3$ & 35.8$^{***}$ & (9.1) & 50.5 \\
$k=4$ & 32.6$^{***}$ & (9.8) & 50.4 \\
$k=5$ & 33.8$^{***}$ & (10.4) & 49.4 \\
$k=6$ & 31.3$^{***}$ & (9.7) & 63.9 \\
$k=7$ & 33.0$^{***}$ & (9.0) & 66.6 \\
$k=8$ & 34.0$^{***}$ & (9.5) & 62.4 \\
$k=9$ & 30.3$^{**}$ & (11.9) & 42.4 \\
$k=10$ & 30.4$^{***}$ & (11.1) & 44.0 \\
\hline
\multicolumn{4}{c}{\textit{Panel B: Instruments built after dropping the top-$k$ industries}} \\
\hline
$k=1$ & 22.8$^{*}$ & (12.6) & 18.9 \\
$k=2$ & 21.8$^{*}$ & (11.9) & 21.6 \\
$k=3$ & 21.5$^{*}$ & (11.7) & 24.5 \\
$k=4$ & 24.6$^{**}$ & (10.5) & 26.9 \\
$k=5$ & 24.0$^{**}$ & (10.4) & 28.0 \\
$k=6$ & 24.2$^{**}$ & (11.2) & 22.3 \\
$k=7$ & 19.4 & (13.3) & 16.8 \\
$k=8$ & 14.3 & (15.0) & 18.6 \\
$k=9$ & 19.5 & (23.4) & 9.9 \\
$k=10$ & 14.6 & (32.0) & 3.9 \\
\hline
\end{tabular}
\par\vspace{6pt}
\begin{minipage}{\textwidth}
\scriptsize
The dependent variable is the change in business applications per 100,000 residents from 2005 to 2019. The baseline row uses the exact leave-one-state-out instrument. Panel~A rebuilds the exact instrument using only the top-$k$ industries ranked by absolute Rotemberg weight. Panel~B rebuilds the exact instrument after excluding the top-$k$ industries. All specifications include state fixed effects and baseline controls, use CZ population weights, and report state-clustered standard errors. KP $F$ denotes the Kleibergen--Paap rk Wald $F$-statistic. $^{***}p<0.01$, $^{**}p<0.05$, $^{*}p<0.1$.
\end{minipage}
\end{table}

%% file: tables_duha/tab_pooled_share_iv.tex
\begin{table}[htbp]
\centering
\caption{Pooled Share-IV Estimates Using Top Rotemberg-Weight Industries\label{tab:bartik_pooled_iv}}
\begin{tabular}{llcccc}
\hline\hline
 & & Coefficient & & & Hansen $J$ \\
Shares used & Estimator & on RSH (s.e.) & Joint $F$ & KP $F$ & $p$-value \\
\hline
Top 5 & 2SLS & 26.9$^{***}$ (4.8) & 106.2 & 99.3 & 0.089 \\
Top 5 & LIML & 27.1$^{***}$ (4.9) & 106.2 & 99.3 & 0.088 \\
Top 10 & 2SLS & 28.0$^{***}$ (4.7) & 72.8 & 68.0 & 0.187 \\
Top 10 & LIML & 28.5$^{***}$ (4.8) & 72.8 & 68.0 & 0.186 \\
\hline
\end{tabular}
\par\vspace{6pt}
\begin{minipage}{\textwidth}
\scriptsize
The dependent variable is the change in business applications per 100,000 residents from 2005 to 2019. ``Top 5'' and ``Top 10'' refer to standardized 1990 CZ employment shares for the industries with the five and ten largest absolute Rotemberg weights. ``Joint $F$'' reports the absorbed first-stage joint significance test from the regression of RSH on the listed shares. KP $F$ reports the Kleibergen--Paap rk Wald $F$-statistic. Hansen $J$ $p$-values come from the overidentification test. All specifications include state fixed effects and baseline controls, use CZ population weights, and report state-clustered standard errors. $^{***}p<0.01$, $^{**}p<0.05$, $^{*}p<0.1$.
\end{minipage}
\end{table}

%% file: tables_duha/table_alternative_treatments.tex
\begin{table}[htbp]
\centering
\caption{Alternative Task-Based Treatment Measures\label{tab:alt_treatments}}
{\setlength{\tabcolsep}{4pt}
\begin{tabular}{@{}l*{4}{c}@{}}
\hline\hline
 & (1) & (2) & (3) & (4) \\
 & First Stage & $\Delta$ Bus.\ App.s & $\Delta$ Bus.\ App.s & $\Delta$ Bus.\ App.s \\
 &  & 2005--19 & 2005--10 & 2010--19 \\
\hline
\\[-6pt]
\multicolumn{5}{l}{\textit{O*NET Degree of Automation}} \\[2pt]
IV & 0.932$^{***}$ & & & \\
 & (0.108) & & & \\
[6pt]
Degree of Automation & & 28.2$^{***}$ & 9.0$^{**}$ & 19.2$^{**}$ \\
 & & (7.9) & (4.0) & (8.1) \\
[6pt]
\hline
\multicolumn{5}{l}{\textit{O*NET Importance of Repetition}} \\[2pt]
IV & 0.851$^{***}$ & & & \\
 & (0.075) & & & \\
[6pt]
Importance of Repetition & & 30.9$^{***}$ & 9.8$^{**}$ & 21.1$^{***}$ \\
 & & (7.2) & (4.2) & (8.2) \\
[6pt]
\hline
\multicolumn{5}{l}{\textit{Routine Cognitive Task Intensity}} \\[2pt]
IV & 0.522$^{***}$ & & & \\
 & (0.062) & & & \\
[6pt]
Routine Cognitive & & 50.3$^{***}$ & 16.0$^{**}$ & 34.3$^{***}$ \\
 & & (11.8) & (7.0) & (13.1) \\
[6pt]
\hline
\multicolumn{5}{l}{\textit{Routine Manual Task Intensity}} \\[2pt]
IV & $-$0.945$^{***}$ & & & \\
 & (0.185) & & & \\
[6pt]
Routine Manual & & $-$27.8$^{***}$ & $-$8.8$^{***}$ & $-$19.0$^{**}$ \\
 & & (8.8) & (3.0) & (9.4) \\
\hline
Controls & Yes & Yes & Yes & Yes \\
State FE & Yes & Yes & Yes & Yes \\
Observations & 722 & 722 & 722 & 722 \\
\hline\hline
\end{tabular}}
\par\vspace{6pt}
\begin{minipage}{\textwidth}
\scriptsize
2SLS estimates. Column~(1) reports first-stage coefficients; columns~(2)--(4) report second-stage estimates. The dependent variable is the change in business applications per 100{,}000 residents. Each panel uses a different task-based treatment measure instrumented with the corresponding 1990 shift-share instrument. All specifications include state fixed effects and baseline commuting-zone controls. Standard errors clustered by state in parentheses. Observations weighted by commuting-zone population share. Coefficient differences across panels mainly reflect different first-stage loadings. Kleibergen--Paap $F$-statistics: Degree of Automation, 74.1; Importance of Repetition, 128.3; Routine Cognitive, 70.8; Routine Manual, 26.0. $^{***}$ $p<0.01$, $^{**}$ $p<0.05$, $^{*}$ $p<0.1$.
\end{minipage}
\end{table}

%% file: tables_duha/table_contract_complexity.tex
\begin{table}[htbp]
\centering
\caption{Contract Complexity and Vendor Diversity in SEC Filings}
\label{tab:contract_complexity}
\begin{threeparttable}
\begin{tabular}{lcccc}
\toprule
 & (1) & (2) & (3) & (4) \\
 & Log Word Count & Log Sections & Nunn Intensity & 1\{Tier-2 KW\} \\
\midrule
\multicolumn{5}{l}{\textit{Panel A: B2B $\times$ Year Trend (Firm + Year FE)}} \\[3pt]
B2B $\times$ Year & -0.00012$^{***}$ & -0.00018$^{***}$ & 0.00045$^{***}$ & 0.00008$^{***}$ \\
  & (0.00002) & (0.00003) & (0.00006) & (0.00001) \\
Firm FE & Yes & Yes & Yes & Yes \\
Year FE & Yes & Yes & Yes & Yes \\
Observations & 21135 & 18085 & 21135 & 21135 \\
Firms & 2304 & 2304 & 2304 & 2304 \\[6pt]
\multicolumn{5}{l}{\textit{Panel B: B2B $\times$ Year Trend by Contract Type (Firm + Year FE)}} \\[3pt]
Service Agreements & -0.00025$^{***}$ & -0.00045$^{***}$ & 0.00058$^{***}$ & 0.00010$^{***}$ \\
  & (0.00004) & (0.00006) & (0.00012) & (0.00001) \\
Supply Contracts & -0.00014$^{***}$ & -0.00013$^{***}$ & 0.00040$^{***}$ & 0.00006$^{***}$ \\
  & (0.00002) & (0.00004) & (0.00007) & (0.00001) \\
License Agreements & 0.00003 & -0.00005 & 0.00045$^{***}$ & 0.00012$^{***}$ \\
  & (0.00003) & (0.00004) & (0.00008) & (0.00001) \\
Observations (Service+Cr) & 16875 & 14473 & 16875 & 16875 \\
Observations (Supply+Cr) & 18643 & 15995 & 18643 & 18643 \\
Observations (License+Cr) & 17698 & 15225 & 17698 & 17698 \\[6pt]
\multicolumn{5}{l}{\textit{Panel C: Within-Firm Vendor Diversity (Firm FE, Year Trend)}} \\[3pt]
 & $\log(1+\text{dyads})$ & Dyads / Contract & \multicolumn{2}{c}{} \\
Year trend & 0.01765$^{***}$ & 0.01434$^{***}$ & & \\
 & (0.00190) & (0.00183) & & \\
Firm FE & Yes & Yes & & \\
Observations & 5628 & 5628 & & \\
Firms & 1360 & 1360 & & \\
\bottomrule
\end{tabular}
\begin{tablenotes}[flushleft]
\small
\item \textit{Notes.} Panels A and B use SEC EDGAR Exhibit 10 material contracts, 2000--2019, with B2B contracts (service, supply, license agreements) as treatment and credit agreements filed by the same firms as control. Panel A reports the differential annual trend for B2B contracts relative to credit agreements. A negative coefficient on Log Word Count says B2B contracts are shrinking faster than credit agreements filed by the same firms. Nunn Intensity is a weighted composite of relationship-specific language density per 1,000 words following Nunn (2007). 1\{Tier-2 KW\} is an indicator for whether the contract contains any of the vendor-management terms \emph{vendor}, \emph{service provider}, \emph{managed service}, \emph{service level}, or \emph{third party}. Panel B repeats the four specifications separately for each B2B contract type, each pooled with the same firms' credit agreements as control. Panels A and B include firm and year fixed effects with standard errors clustered by firm. Panel C reports within-firm year trends in distinct vendor counterparties (dyads) from a firm-year aggregation. The within-firm year trend on the dyad count is positive while the total number of B2B contracts filed is flat, indicating that firms are spreading vendor work across a wider set of outside companies. Panel C uses firm fixed effects only and does not have a credit-agreement control because credit filings have a single counterparty by construction. $^{***}p<0.01$, $^{**}p<0.05$, $^{*}p<0.10$.
\end{tablenotes}
\end{threeparttable}
\end{table}

%% file: tables_duha/table_govt_contract_categories.tex
\begin{table}[htbp]
\centering
\caption{Service Contract Governance by PSC Category, FY2007 vs.\ FY2019\label{tab:govt_contract_categories}}
\small
\begin{tabular}{cl rr rr r rr r}
\hline\hline
 & & \multicolumn{2}{c}{Transactions} & \multicolumn{3}{c}{Performance-Based (\%)} & \multicolumn{3}{c}{Fixed-Price (\%)} \\
\cmidrule(lr){3-4} \cmidrule(lr){5-7} \cmidrule(lr){8-10}
PSC & Category & FY07 & FY19 & FY07 & FY19 & $\Delta$ & FY07 & FY19 & $\Delta$ \\
\hline
L & Tech Representative & 3{,}673 & 2{,}081 & 15.8 & 55.8 & 40.0 & 61.4 & 58.5 & $-$2.9 \\
A & R\&D & 60{,}544 & 36{,}896 & 12.1 & 50.0 & 37.9 & 30.5 & 27.7 & $-$2.8 \\
Z & Maint/Repair Real Prop. & 54{,}904 & 37{,}556 & 3.1 & 40.2 & 37.1 & 96.3 & 98.0 & 1.7 \\
R & Prof/Admin/Mgmt & 162{,}908 & 173{,}263 & 20.1 & 55.4 & 35.3 & 57.8 & 72.7 & 14.9 \\
S & Utilities/Housekeeping & 73{,}413 & 49{,}254 & 21.7 & 55.6 & 33.9 & 87.2 & 92.6 & 5.4 \\
D & IT/Telecom & 55{,}106 & 83{,}965 & 27.9 & 61.1 & 33.2 & 64.1 & 88.7 & 24.6 \\
N & Install Equipment & 7{,}697 & 6{,}050 & 6.4 & 38.6 & 32.2 & 84.9 & 91.5 & 6.6 \\
W & Lease Equipment & 21{,}144 & 12{,}112 & 3.3 & 33.7 & 30.4 & 96.1 & 99.9 & 3.8 \\
J & Maint/Repair Equipment & 106{,}038 & 55{,}028 & 10.3 & 40.3 & 30.0 & 84.6 & 90.4 & 5.7 \\
H & Quality Control & 5{,}831 & 5{,}682 & 12.6 & 41.4 & 28.8 & 82.6 & 89.3 & 6.7 \\
G & Social Services & 4{,}679 & 3{,}412 & 11.2 & 39.9 & 28.6 & 89.7 & 92.3 & 2.7 \\
Q & Medical Services & 38{,}854 & 30{,}354 & 8.2 & 34.4 & 26.3 & 87.7 & 98.5 & 10.7 \\
B & Special Studies & 12{,}902 & 6{,}654 & 17.2 & 42.8 & 25.6 & 64.7 & 66.5 & 1.8 \\
F & Natural Resources & 17{,}899 & 19{,}386 & 16.9 & 42.1 & 25.2 & 75.3 & 88.2 & 13.0 \\
M & Operate Facilities & 5{,}288 & 5{,}053 & 48.3 & 72.6 & 24.3 & 47.9 & 57.8 & 9.9 \\
V & Transport/Travel & 24{,}527 & 15{,}455 & 22.0 & 43.6 & 21.5 & 87.1 & 98.9 & 11.8 \\
T & Photo/Mapping/Printing & 5{,}456 & 1{,}808 & 8.7 & 29.8 & 21.1 & 87.2 & 93.3 & 6.1 \\
X & Lease Facilities & 110{,}390 & 5{,}104 & 0.4 & 21.4 & 21.0 & 99.4 & 99.8 & 0.4 \\
K & Mod of Equipment & 2{,}166 & 1{,}826 & 32.5 & 53.1 & 20.6 & 36.6 & 42.7 & 6.1 \\
U & Education/Training & 28{,}173 & 13{,}267 & 29.8 & 39.8 & 10.0 & 87.3 & 94.0 & 6.7 \\
C & Architect/Engineering & 24{,}325 & 15{,}478 & 0.7 & 10.3 & 9.6 & 86.3 & 92.0 & 5.7 \\
Y & Construction & 22{,}693 & 16{,}086 & 0.2 & 6.8 & 6.6 & 96.5 & 99.2 & 2.7 \\
\hline\hline
\end{tabular}
\par\vspace{6pt}
\begin{minipage}{\textwidth}
\scriptsize
Source is FPDS/USAspending. Sample restricted to domestic, positive-obligation award transactions with service PSC codes (letter prefix). Performance-based acquisition indicates contracts specifying outcomes rather than processes (FAR Subpart 37.6). Fixed-price includes firm fixed price (J) and fixed price with economic price adjustment (K). Rows sorted by performance-based share change, largest at top. Categories with the largest shifts toward performance-based contracting (professional/administrative support, IT, facility maintenance) correspond to codifiable, monitorable service tasks. Categories with the smallest shifts (construction, architecture/engineering) involve bespoke, site-specific work that resists outcome-based specification.
\end{minipage}
\end{table}